\documentclass[12pt]{article}
\usepackage{amsmath}
\usepackage{graphicx}
\usepackage{amsfonts}										
\usepackage{amsthm}
\usepackage{enumerate}
\usepackage{multicol}
\usepackage[titletoc]{appendix}
\usepackage{multirow}
\usepackage{booktabs}
\usepackage{fancyhdr}
\usepackage{epsfig}
\usepackage{nomencl}
\usepackage[authoryear]{natbib}
\usepackage[titletoc]{appendix}
\usepackage[T1]{fontenc}
\usepackage[utf8]{inputenc} 
\usepackage{boiboites}
\usepackage{leftidx}
\usepackage{setspace}  
\usepackage{fourier}
\usepackage[english]{babel}															
\usepackage[protrusion=true,expansion=true]{microtype}				
\usepackage{multirow}
\usepackage{lipsum}
\usepackage{caption}											
\usepackage{url} 

\usepackage[colorlinks,citecolor=blue,urlcolor=blue,bookmarks=false,hypertexnames=true]{hyperref}

 \makeatletter
 \g@addto@macro\normalsize{%
 	\setlength\abovedisplayskip{8pt}
 	\setlength\belowdisplayskip{8pt}
 	\setlength\abovedisplayshortskip{8pt}
 	\setlength\belowdisplayshortskip{8pt}
 }

\newcommand{\blind}{1}

\newtheorem{theoremm}{Theorem}

\newboxedtheorem[boxcolor=black, background=green!10, titlebackground=white,titleboxcolor = black]{theom}{Theorem}{thCounter}
\newboxedtheorem[boxcolor=black, background=green!10, titlebackground=white,titleboxcolor = black]{defi}{Definition}{a}
\newboxedtheorem[boxcolor=black, background=green!10, titlebackground=white,titleboxcolor = black]{result}{Result}{b}
\newboxedtheorem[boxcolor=black, background=green!10, titlebackground=white,titleboxcolor = black]{lemma}{Lemma}{c}
\newboxedtheorem[boxcolor=black, background=green!10, titlebackground=white,titleboxcolor = black]{property}{Properties}{d}
\newcommand{\E}{\operatorname{E}}

\newcommand{\var}{\operatorname{Var}}

\newcommand{\mse}{\operatorname{MSE}}

\newcommand{\diag}{\operatorname{diag}}

\newcommand{\bb}{\mathbf{B}}
\newcommand{\bc}{\mathbf{C}}

\newcommand{\bk}{\mathbf{K}}

\newcommand{\bv}{\mathbf{V}}

\newcommand{\ssb}{\mathbf{b}} 

\newcommand{\su}{\mathbf{u}}
\newcommand{\sv}{\mathbf{v}}

\newcommand{\sx}{\mathbf{x}}
\newcommand{\sz}{\mathbf{z}}

\newcommand{\btheta}{\boldsymbol{\theta}}
\newcommand{\bomega}{\boldsymbol{\omega}}

\newcommand{\bbeta}{\boldsymbol{\beta}}

\newcommand{\bzero}{\boldsymbol{0}}
\newcommand{\bone}{\boldsymbol{1}}
\newcommand{\bxi}{\boldsymbol{\xi}}

\newcommand{\hbomega}{\hat{\bomega}}
\newcommand{\hmse}{\widehat{\mse}}

\newcommand{\dbeta}{\dot{\beta}}

\newcommand{\dsiga}{{\dot{\sigma}_{\alpha}}}
\newcommand{\dsige}{{\dot{\sigma}_e}}

\newcommand{\yij}{y_{ij}}
\newcommand{\eij}{e_{ij}}

\newcommand{\sumig}{\sum_{i=1}^g}

\newcommand{\hsige}{{{\hat{\sigma}_e}}}

\newcommand{\sige}{\sigma_e}
\newcommand{\siga}{\sigma_{\alpha}}

\numberwithin{equation}{section}

\begin{document}
\def\spacingset#1{\renewcommand{\baselinestretch}%
	{#1}\small\normalsize} \spacingset{1}

\if1\blind
{
\title{Small Area Estimation using EBLUPs under the Nested Error Regression Model}
\author{ Ziyang Lyu  
	\thanks{Ziyang Lyu is a  postdoctoral research fellow  in 
		UNSW Data Science Hub, School of Mathematics and Statisics,University of New South Wales, NSW, 2033, Australia.  Email: \textsf{lvziyang08@gmail.com}. A.H.Welsh is E.J. Hannan Professor of Statistics in the Research School of Finance, Actuarial Studies and Statistics, Australian National University,  ACT, 2601, Australia.  Email: \textsf {Alan.Welsh@anu.edu.au}.}
\hspace{.2cm} \\
UNSW Data Science Hub, 
School of Mathematics and Statisics
\\University of New South Wales\\
A.H. Welsh  \\
Research School of Finance, Actuarial Studies and Statistics, \\  Australian National University}

	\maketitle
} \fi

\if0\blind
{
	\bigskip
	\bigskip
	\bigskip
	\begin{center}
		{\LARGE\bf Small Area Estimation using EBLUPs under the Nested Error Regression Model }
	\end{center}
	\medskip
} \fi

	\bigskip
	\begin{abstract}
Estimating characteristics of domains (referred to as small areas) within a population from sample surveys of the population is an important problem in survey statistics.  In this paper, we consider model-based small area estimation under the nested error regression model.  We discuss the construction of mixed model estimators (empirical best linear unbiased predictors, EBLUPs) of small area means and the conditional linear predictors of small area means.  Under the asymptotic framework of increasing numbers of small areas and increasing numbers of units in each area, we establish asymptotic linearity results and central limit theorems for these estimators which allow us to establish asymptotic equivalences between estimators,  approximate their sampling distributions, obtain simple expressions for and construct simple estimators of their asymptotic mean squared errors, and justify asymptotic prediction intervals.   We present model-based simulations that show that in quite small, finite samples, our mean squared error estimator performs as well or better  than the widely-used \cite{prasad1990estimation} estimator and is much simpler, so is easier to interpret.   We also carry out a design-based simulation using real data on consumer expenditure on fresh milk products to explore the design-based properties of the mixed model estimators.  We explain and interpret some surprising simulation results through analysis of the population and further design-based simulations.  The simulations highlight important differences between the model- and design-based properties of mixed model estimators in small area estimation.
	\end{abstract}
	
	\noindent%
	\vfill
	
 \newpage
 \spacingset{1.8} 
	\section{Introduction}

Estimates of area-level characteristics of interest (such as means, totals, and quantiles) for areas, domains or clusters within a population (all intended to be included whenever we refer to areas) obtained from sample survey data are widely used for resource allocation in social, education and environmental programs, and as the basis for commercial decisions.  Direct estimates which use only data specific to an area, can have large standard errors because of relatively small area-specific sample sizes.  Small area estimation is concerned with producing more reliable estimates with valid measures of uncertainty for the characteristics of interest; recent reviews have been given by for example \cite{rao2005inferential,rao2008some}, \cite{lehtonen2009design},  \cite{pfeffermann2013new}, and \cite{rao2014small}.   

A popular method for obtaining reliable estimates  \citep{fay1979, battese1988error} is to introduce a mixed model for the population which includes fixed effects (to describe either unit-level or area-level effects) and random effects (to capture additional between area variation), fit the model using data from multiple areas and then, use the fitted model to construct the desired estimates.  When we have unit-level data, a simple and widely used model is the nested error regression or random intercept model \citep{battese1988error}, and a widely used method for estimating means or totals is to use empirical best linear unbiased predictors (EBLUPs) obtained by minimising the (prediction) mean squared error and then estimating the unknown quantities by maximum likelihood or restricted maximum likelihood (REML) estimation; see for example \cite{saei2003a,saei2003},  \cite{jiang2006estimation} and \cite{haslett2019}. In addition to the issue of the level at which the data are available, there is also a subtlety about the target of estimation.  The most commonly studied characteristics of interest are small area means (or equivalently totals) and, when we assume a mixed model, the conditional expectations of the small area means given the random effects.  These two targets, the small area means and their conditional expectations are different and have different EBLUPs with potentially different mean squared errors, but are treated as interchangeable in small area estimation.  They are both random variables under the model-based framework, so technically they need to be predicted rather than estimated.  However, it is common to use both ``prediction'' and ``estimation'' in small area estimation and we refer to their EBLUPs as mixed model estimators that are distinguished by their different targets  \citep{tzavidis2010robust}; they can also be described as composite and synthetic estimators respectively.

The model-based variability of small area estimators is usually described by reporting estimates of their (prediction) mean squared errors or by prediction intervals, often based on these estimates.  Estimation of (prediction) mean squared errors for mixed model estimators is complicated, even for simple linear mixed models like the nested error regression model:  estimates of (prediction) mean squared error for mixed model estimators based on treating the variance parameters as known (i.e. for the BLUPs rather than the EBLUPs) are underestimates when linear mixed models are fitted to real data; and simple, analytic expressions for the (prediction) mean squared errors of mixed model estimators are not available, complicating their estimation.  Under normal linear mixed models (including the nested error regression model), when the number of areas is allowed to increase while the area sizes are held fixed (or bounded), \cite{kackar1981unbiasedness} and \cite{prasad1990estimation} used second order Taylor expansions to obtain approximations to the (prediction) mean squared error of the EBLUPs of the conditional expectation of the small area means, and then constructed mean squared error estimators by replacing the unknown quantities in these approximations by estimators.  The Prasad-Rao approximation and estimator have been extended to more general models and to allow additional estimators of the model parameters by \cite{datta2000unified} and \cite{das2004mean};  see also \cite{zadlo2009} and \cite{torabi2013estimation}.  Alternatives to estimators based on analytic approximations include estimators obtained using resampling methods.  \cite{jiang2002unified} proposed and investigated cluster-level jackknife methods (unusually, treating the small area means as the characteristics of interest), \cite{hall2006a} proposed a parametric bootstrap approach for constructing bias-corrected estimates of the prediction mean squared error and prediction regions and \cite{chatterjee2008} used a different approach to construct parametric bootstrap prediction intervals. For the considerably more complicated non-normal case,  \cite{hall2006b} proposed a moment-matching, double-bootstrap procedure to estimate the prediction mean squared error. 

Many extensions of the basic approach have been developed.  These include, for example, introducing 
outlier robust estimators \citep{sinha2009robust, tzavidis2010robust}, spatial models to allow correlation between small areas \citep{saei2005, torabi2020estimation}, and different response distributions through using generalized linear mixed models \citep{saei2003a} and models that allow responses with extra zeros \citep{chandra2016small}.  Other work has incorporated design-based considerations \citep{jiang2006estimation, jiang2011best}.


The standard asymptotic framework for model-based small area estimation under the nested error regression model follows \cite{kackar1981unbiasedness} and \cite{prasad1990estimation} in allowing the number of areas to increase while holding the area sizes fixed (or bounded).  This framework has several disadvantages, most notably that the small area estimators are not consistent and their asymptotic distributions are not known.  In turn, this hinders the derivation of approximate distributions (none are available), complicates the construction of mean squared error estimates and means that prediction intervals based on the estimated mean squared errors cannot be shown to achieve their nominal level even asymptotically.  To overcome these difficulties, we need both the number of areas and the sample size in each area to increase.  This appears to contradict the ``small'' in ``small area estimation'', but the approximations derived within this framework perform well (the ultimate purpose of deriving asymptotic approximations) even when some areas have quite small sample size.  In addition, many practical applications include a number of large ``small areas'' and no tiny ones, so it is often a practically relevant framework.  Examples occur in clinical research (clustered trials) when we study records on large groups (areas) of patients (units) with each group treated by a different medical practitioner or at a different hospital, in educational research when we look at records on college students (units) grouped within schools (areas), and in sample surveys when we observe people or households (units) grouped in defined clusters (areas). For example, \cite{arora1997superiority} gave an instance with $43$ areas ranging in size from $95$ to $633$ units, and such examples are common in poverty data \citep{pratesi2016analysis}.

We use recent increasing number of areas and increasing area size asymptotic results obtained by \cite{lyu2022increasing} for estimating the parameters of the nested error regression model by maximum likelihood or REML estimation and by  \cite{lyu2021asymptotics} for the EBLUPs for the random effects in the nested error regression model, to study small area estimators in the same framework.  Without having to assume normality in the model, we obtain simple approximations to the distributions of the mixed model estimators that involve the distribution of the characteristic of interest and a normal distribution.  We further obtain strikingly simple expressions for the asymptotic mean squared errors of the estimators which are easy to estimate and can be used in prediction intervals which are demonstrated to have correct asymptotic coverage.  Such results are are not available under the standard asymptotic framework. They are achieved by approximating the estimators directly and taking the (prediction) mean squared error of the approximation, rather than directly approximating the (prediction) mean squared error. Our results fill the practical and theoretical gap around the most widely used mixed model estimators in small area estimation and suggest how other similar gaps can be treated.   

We describe the nested error regression model, discuss the targets of estimation and the mixed model estimators we consider in Section   \ref{sec:desc}.  We present our increasing number of areas and increasing area size asymptotic results in Section \ref{sec:asympts} and use (model-based) simulation to demonstrate the relevance of these results to finite samples in Section \ref{sec:sims}.  We include a design-based simulation using real data on consumer expenditure on fresh milk products, and then use additional design-based simulations to explore some unexpected findings in Section \ref{sec:example}.  We conclude with a brief discussion in Section \ref{sec:discussion}.

\section{Small area estimation}\label{sec:desc}

Consider a population $U=\cup_{i=1}^g U_i$ of $N$ units, partitioned into $g$ exclusive areas  $U_i$, each containing $N_i$ units so that $\sum_{i=1}^g N_i=N$. 
Let $y_{ij}$ be a scalar survey variable of interest and $\sx_{ij}$ a vector of auxiliary variables for the $j$th unit in the $i$th area.  The problem of interest is to use data from a sample of units in $U$ to make inference about the area means $\bar{y}_i = N_i^{-1}\sum_{j=1}^{N_i} y_{ij}$.   We assume that the values of the auxiliary variables are known for every unit in the population and that the values of the survey variable are observed on a sample of units $s=\cup_{i=1}^g s_i$, where $s_i \subseteq U_i$ of size $n_i \le N_i$  is the set of sample units within $U_i$ and $\sum_{i=1}^g n_i=n$.   We assume that the units are selected using a non-informative sampling method such that $n_i > n_L >0$, so that some units from every area are included in the sample.  (This can be weakened provided the areas with no units are a non-informatively selected sample of the areas.) The available data $\mathcal{D}$ consists of $y_{ij}$ for $j \in s_i$, $i=1,\ldots,g$ and $\sx_{ij}$ for $j =1,\ldots, N_i$, $i=1,\ldots,g$.  

We assume the nested error regression model (or random intercept model) for the survey variable at the population level, so
	\begin{equation}\label{nerm}
		\yij=\mu(\sx_{ij})+\alpha_i+\eij, \qquad\text{for $ j=1,\ldots,N_i, \, i=1,\ldots,g$,}
	\end{equation}
where $\mu(\sx_{ij})$ is the regression function (the conditional mean of the response) given $\sx_{ij}$,  $\alpha_i$ is a random effect representing a random intercept or cluster effect and $e_{ij}$ is a random error.  We assume that the $\{\alpha_i\}$ and $\{e_{ij}\}$ are all mutually independent with mean zero and variances (called the variance components) $\siga^2$ and  $\sige^2$, respectively, and write $\btheta=[\siga^2, \sige^2]^T$.   These random variables do not have to be normally distributed so the responses $y_i$ are not necessarily normally distributed.  \cite{lyu2022increasing} showed that, in asymptotic theory with increasing area size, we need to distinguish between area variables (variables that vary within areas) which we place in the $p_w$-vector $\sx_{ij}^{(w)}$ and the between area variables (variables that are constant within areas) which we place in the $p_b$-vector $\sx_{i}^{(b)}$.   We then write the regression function as
	\begin{equation}\label{mean}
		\mu(\sx_{ij}) = \beta_0+\sx_{i}^{(b)T}\bbeta_1+ \sx_{ij}^{(w)T}\bbeta_2 = \su_i^T\bxi + \sx_{ij}^{(w)T}\bbeta_2,
	\end{equation}
where $\beta_0$ is the unknown intercept, $\bbeta_1$ is the unknown between area slope parameter and $\bbeta_2$ is the unknown within area slope parameter.  The intercept is a between area variable so it is convenient to group the intercept and the between area slope terms in the $(p_b+1)-$vectors $\su_i = [1,\,\, \sx_{i}^{(b)T}]^T$ and $\bxi= [\beta_0, \bbeta_1^T]^T$.  It is also sometimes useful to write the regression function in terms of $\sz_{ij} = [1,\,\, \sx_{i}^{(b)T}, \,\, \sx_{ij}^{(w)T}]^T =  [\su_{i}^{T}, \,\, \sx_{ij}^{(w)T}]^T$ and $\bbeta = [\beta_0, \bbeta_1^T, \bbeta_2^T]^T=[\bxi^T,\bbeta_2^T]^T$.   Finally, grouping the between and within area parameters, we write the full set of model parameters as $\bomega=[\beta_0,\bbeta_1^T,\siga^2,\bbeta_2^T,\sige^2]^T$.

The form (\ref{mean}) for the regression function includes the three cases with either or both between and within area variables.  It also gives us the option of replacing $\sx_{ij}^{(w)}$ by the population-area-mean-centered within area variables $\sx_{ij}^{(w)}-\bar{\sx}_i^{(w)}$, where $\bar{\sx}_i^{(w)} = N_i^{-1}\sum_{k=1}^{N_i} \sx_{ik}^{(w)} $, for  $j=1,\ldots,N_i$, $i=1,\ldots,g$, and including the area means $\bar{\sx}_i^{(w)}$ as contextual effects with the between area variables from $\sx_{ij}$ in $\sx_{i}^{(b)}$.  This centering makes the between and within area variables orthogonal and allows a very simple interpretation of the parameters from separate between and within area models; see \cite{yoon2020effect} and the references therein.  The maximum likelihood and REML estimators of the parameters are asymptotically orthogonal under the conditions of Theorem1 of \cite{lyu2022increasing}, so the effect is on the model parameters rather than the estimators; this is the reason we center about the population area means rather than the sample area means.  


	\newcommand{\xbariw}{\bar{\sx}_{i}^{(w)}}
	\newcommand{\xbarisw}{\bar{\sx}_{i(s)}^{(w)}}
	\newcommand{\xbarirw}{\bar{\sx}_{i(r)}^{(w)}}
	\newcommand{\xbariwT}{\bar{\sx}_{i}^{(w)T}}
	\newcommand{\xbariswT}{\bar{\sx}_{i(s)}^{(w)T}}
	\newcommand{\xbarirwT}{\bar{\sx}_{i(r)}^{(w)T}}
	\newcommand{ \ysyn}{M_i^{\text{clp}}}
	\newcommand{\ycom}{M_i^{\text{sam}}}
	\newcommand{ \hysyn}{\hat{M}_i^{\text{clp}}}
	\newcommand{\hycom}{\hat{M}_i^{\text{sam}}}

The two common targets when we are interested in the small area means are the actual small area means $\bar{y}_i = \su_i^T\bxi+\bar{\sx}_i^{(w)^T}\bbeta_{2}+\alpha_i+\bar{e}_{i}$
and the conditional linear predictors of the small area means 
\begin{equation*}\label{p-target2}
	{\eta}_i = \E(\bar{y}_i|\su_i, \sx_{i1}^{w},\ldots, \sx_{iN_i}^{(w)}, \alpha_i ) 
	= \su_i^T{\bxi}+\xbariwT{\bbeta}_{2}+\alpha_i,
\end{equation*}
$i=1,\ldots,g$.  The two targets satisfy $\bar{y}_i - \eta_i = \bar{e}_i$ so are not identical and the differences do not decrease in the standard asymptotic framework with fixed area size; the targets have the same expectations $\E(\bar{y}_i) = \E(\eta_i)$ but different variances $\var(\bar{y}_i) = \sigma_{\alpha}^2 + N_i^{-1}\sigma_{e}^2$ and $\var(\eta_i) = \sigma_{\alpha}^2$.  This has practical consequences because unbiased predictors of one will also be unbiased predictors of the other, but the (prediction) mean squared errors of unbiased predictors will differ with the target.  We prefer to predict $\bar{y}_i$ rather than $\eta_i$ because $\bar{y}_i$ is an easily interpretable finite population parameter that is is not tied to a specific population model and can be evaluated without any prediction error when $U_i$ is completely enumerated.  In contrast, $\eta_i$ is more difficult to interpret, is not a finite population parameter, is tied to the specific population model we are using and is subject to prediction error even when $U_i$ is completely enumerated.  The use of $\eta_i$ as a target for prediction in small area estimation arguably reflects the fact that early work on prediction for mixed models like (\ref{nerm}) and (\ref{mean}) \citep{kackar1981unbiasedness, prasad1990estimation} which was brought into the finite population context for small area estimation was actually done in an infinite population framework where $\bar{y}_i$ was observed (so was a trivial target)) and prediction was to other (infinite) populations (assumed to follow the same model) rather than within a given finite population.   In finite populations, prediction within the given population is the more relevant problem.  One advantage of working under our increasing area size framework is that the two targets are asymptotically the same up to order $N_i^{-1/2}$, because
\[
\bar{y}_i - \eta_i = \bar{e}_i = O_p(N_i^{-1/2}), \qquad \mbox{ as } N_i \rightarrow \infty.
\]
This means that the predictors are quite similar in large areas and may explain why the distinction between the two targets has been largely ignored in practice.

The prediction mean squared error for predicting a target random variable is minimised by the conditional expectation of the target given the observed data $\mathcal{D}$. Typically, the distribution of the target given $\mathcal{D}$ is derived from a model with unknown parameters (such as (\ref{nerm}) and (\ref{mean})) so the conditional expectation depends on these unknown parameters and a feasible predictor requires replacing the unknown parameters by estimators; in our case, we use the maximum likelihood or REML estimators of the parameters in the model (\ref{nerm}) and (\ref{mean}). 
		
To simplify notation, we write $j \not\in s_i$ to mean $j \in U_i\backslash s_i$, $k_i=(N_i-n_i)/N_i$, and use subscripts $(s)$ and $(r)$ to denote quantities related to sampled and non-sampled units. Specifically, let $\bar{y}_{i(s)}=n_i^{-1}\sum_{j \in s_i}y_{ij}$,  $\bar{y}_{i(r)}=(N_i-n_i)^{-1}\sum_{j\not\in s_i}y_{ij}$,  $\xbarisw=n_i^{-1}\sum_{j\in s_i}\sx_{ij}^{(w)}$, $\xbarirw=(N_i-n_i)^{-1}\sum_{j\not\in s_i}\sx_{ij}^{(w)}$,    $\bar{e}_{i(s)}=n_i^{-1}\sum_{j \in s_i}e_{ij}$ and $\bar{e}_{i(r)}=(N_i-n_i)^{-1}\sum_{j\not\in s_i}e_{ij}$.  Then, under the model (\ref{nerm}) and (\ref{mean}),  we can write the actual small area means $\bar{y}_i$ as
\begin{equation}\label{p-target1}
	\bar{y}_{i}=(1-k_i)\bar{y}_{i(s)}+k_i\bar{y}_{i(r)}=(1-k_i)\bar{y}_{i(s)}+k_i(\su_i^T\bxi+\xbarirwT\bbeta_{2}+\alpha_i+\bar{e}_{i(r)}).
\end{equation}
Taking the conditional expectation given the data $\mathcal{D}$ of (\ref{p-target1}) and
substituting the maximum likelihood or restricted maximum likelihood (REML) estimators	
$\hat{\bxi}$ and $\hat{\bbeta}_{2}$ for $\bxi$ and $\bbeta_2$,   respectively, and predicting $\E(\alpha_i|\mathcal{D})$ by the empirical best linear unbiased predictor (EBLUP)
	\begin{equation} \label{blup}
		\hat{\alpha}_i= \hat{\gamma}_i \{\bar{y}_{i(s)} - \su_i^T\hat{\bxi} - \xbariswT\hat{\bbeta}_{2}\}, \qquad \mbox{ with } \,\hat\gamma_i = s_i\hat\sigma_{\alpha}^2/(\hat\sigma_{e}^2+s_i\hat\sigma_{\alpha}^2),
	\end{equation}
we obtain the predictor 
\begin{equation}\label{composite}
		\begin{split}
			\hat{M}_i^{sam}&=(1-k_i)\bar{y}_{i(s)}+k_i\{ \su_i^T\hat{\bxi}+\xbarirwT\hat{\bbeta}_{2}+\hat{\alpha}_i\}, 
		\end{split}
\end{equation}
which the mixed model estimator of $\bar{y}_i$  \citep{tzavidis2010robust}; the superscript 'sam' shows that the target is the `small area mean'. It can also be called an EBLUP or a composite estimator \citep{costa2003using} as it combines the sample mean $\bar{y}_{i(s)}$ with a synthetic component $\su_i^T\hat{\bxi}+\xbarirwT\hat{\bbeta}_{2}+\hat{\alpha}_i$ \citep{rao2008some}.
The mixed model estimator of $\eta_i$ obtained as above is
	\begin{equation}\label{syn}
	\begin{split}
		\hat{M}_i^{\text{clp}}&=\su^T_i\hat{\bxi}+\bar{\sx}^{(w)T}_i\hat{\bbeta}_{2}+\hat{\alpha}_i
		= (1-k_i)\{ \su_i^T\hat{\bxi}+\xbariswT\hat{\bbeta}_{2}+\hat{\alpha}_i\}+k_i\{ \su_i^T\hat{\bxi}+\xbarirwT\hat{\bbeta}_{2}+\hat{\alpha}_i\},
	\end{split}
	\end{equation}
which is an EBLUP and a fully synthetic or indirect estimator \citep{prasad1990estimation,lahiri1995robust, jiang2011best}; the superscript `clp' shows that the target is the `conditional linear predictor'.

The main difference between the mixed model estimator of $\bar{y}_i$ (\ref{composite}) and the mixed model estimator of $\eta_i$ (\ref{syn}) estimators is that the former uses the observed $\bar{y}_{i(s)}$ whereas the latter uses a model-based prediction for this quantity.    The difference between the estimators of the two targets can be expressed in terms of the EBLUP for the random effect (\ref{blup}) as
	\begin{equation}\label{diff-syn-com}
		\begin{split}
			\hycom-\hysyn&=(1-k_i)\{ \bar{y}_{i(s)} -  \su_i^T\hat{\bxi}-\xbariswT\hat{\bbeta}_{2}-\hat{\alpha}_i\}
			=\frac{n_i}{N_i}\left\lbrace  \frac{\hat{\alpha}_i}{\hat{\gamma}_i}-\hat{\alpha}_i\right\rbrace 
			=\frac{\hat{\sigma}_e^2}{\hat{\sigma}_{\alpha}^2}\frac{\hat{\alpha}_i}{N_i}.
		\end{split}
	\end{equation} 
This difference is often quite small, but it can be large for the areas with extreme EBLUPs $\hat{\alpha}_i$, particularly if these are small areas and the estimated within area variance is much larger than the estimated between area variance so $\hat{\sigma}_e^2 > \hat{\sigma}_{\alpha}^2$.  
Asymptotically, the difference is $O_p(N_i^{-1})$, so the estimators are asymptotically equivalent up to this order and hence asymptotically closer than their respective targets; see Section \ref{sec:asympts} for details.  Again, these properties only hold when the area sizes are increasing and do not hold for fixed area-size asymptotics.

	We have emphasised that the predictor depends on the target random variable, but of course it also depends on the data and the model.  We note that both predictors (\ref{composite}) and (\ref{syn}) can be computed from the sample data $\mathcal{D}_{(s)}=\{(y_{ij}, \sx_{ij}^{(w)T})^T, j=1,\ldots n_i, \su_i, i=1,\ldots, g\}$ and the the population means $\xbariw$, $ i=1,\ldots, g$. If, instead of $\mathcal{D}$,  we only observe $\mathcal{D}_{(s)}$, the population mean $\xbariw$ and hence the non-sample mean of the within area variables $\xbarirw$ is unknown and also needs to be predicted.
If we estimate $\xbarirw$ by the simple nonparametric estimator $\xbarisw$,  we obtain the mixed model estimator/predictor of the small area mean $\bar{y}_i$ 
\begin{equation*}\label{compositestar}
		\begin{split}
			\hat{M}_i^{*sam}&=(1-k_i)\bar{y}_{i(s)}+k_i\{ \su_i^T\hat{\bxi}+\xbariswT\hat{\bbeta}_{2}+\hat{\alpha}_i\}. 
		\end{split}
\end{equation*}
Clearly, precision in the description of predictors (especially optimal predictors) is needed to avoid confusion.

The asymptotic results (especially Therem \ref{corol1})  in Section \ref{sec:asympts} below establish that, as $g, n_L\to\infty$, the mixed model estimators are asymptotically equivalent predictors of the small area means $\bar{y}_i$, they are asymptotically unbiased predictors and their asymptotic prediction mean squared errors are $\mse_{\text{LW},i} =n_i^{-1}k_i\sige^2$.   We can estimate $\mse_{\text{LW},i}$ by substituting the maximum likelihood or REML estimator for $\sige^2$ to obtain
	\begin{eqnarray}
		&	\hmse_{\text{LW},i} =n_i^{-1}k_i\hsige^2.\label{hmse:com}
	\end{eqnarray}
We can then construct simple, asymptotic $100(1-\varepsilon)\%$ prediction intervals for  $\bar{y}_{i}$ which we denote sam-LW and clp-LW, respectively, as
	\begin{eqnarray}
		&	[\hycom- \Phi^{-1}(1-\varepsilon/2)	\hmse_{\text{LW},i}^{1/2}, \, \hycom+ \Phi^{-1}(1-\varepsilon/2)	\hmse_{\text{LW},i}^{1/2}],\label{predint1}\\
		&	[ \hysyn- \Phi^{-1}(1-\varepsilon/2)	\hmse_{\text{LW},i}^{1/2}, \,  \hysyn+ \Phi^{-1}(1-\varepsilon/2)	\hmse_{\text{LW},i}^{1/2}],\label{predint2}
	\end{eqnarray}
where $\Phi^{-1}$ is the inverse of the standard normal cumulative distribution function.  The asymptotic coverage of the intervals (\ref{predint1}) and (\ref{predint2}) is guaranteed by Theorem \ref{thm1}, Slutsky's Theorem and Theorem \ref{corol1}, which do not require the assumption of normality in the model.

\section{Increasing area-size asymptotic results} \label{sec:asympts}


	We assume throughout that the true model describing the actual data generating mechanism is  given by (\ref{nerm}) and (\ref{mean}) with true parameter $\dot\bomega=[\dbeta_0,\dot\bbeta_1^T,\dsiga^2,\dot\bbeta_2^T,\dsige^2]^T$  and  take all expectations under the true model.  Let $\hat{\bomega}$ denote the normal maximum likelihood estimator  (MLE) of  $\dot\bomega$ obtained by maximizing the normal likelihood based on a sample $s$.  Similarly, let $\hat{\btheta}_{R}$ be the normal REML estimator of $\dot\btheta = [\dsiga^2,\dsige^2]^T$ obtained by maximizing the normal REML criterion function and  let $\hat{\bbeta}(\btheta) = [\hat{\beta}_0(\btheta), \hat{ \bbeta}_1(\btheta)^T, \hat{\bbeta}_2(\btheta)^T]^T$ be the profile likelihood estimator of $\dot\bbeta = [\dbeta_0,\dot\bbeta_1^T,\dot\bbeta_2^T]^T$ obtained by maximizing the normal likelihood with $\btheta$ held fixed.   We call $\hat{\bbeta}_{R} = \hat{\bbeta}(\hat{\btheta}_R) $ the normal REML estimator of $\dot\bbeta$, and $\hbomega_{R} = (\hat{\beta}_{R0}, \hat{\bbeta}_{R1}^T, \hat{\sigma}_{R\alpha}^2, \hat{\bbeta}_{R2}^T, \hat{\sigma}_{Re}^2)^T$  the normal REML estimator of $\dot\bomega$.  
	
	
Following \cite{lyu2022increasing},  we obtain the asymptotic properties of the maximum likelihood and REML estimators of the model parameters under  the following conditions: 
	
	\bigskip\noindent
	\textbf{Condition A}
	\begin{enumerate}
		\item The model (\ref{nerm}) and (\ref{mean}) holds with true parameters $\dot\bomega$ inside the parameter space $\Omega$.
		\item The number of clusters $g \to \infty$ and the minimum area sample size $n_{L}\to\infty$. 
		\item   The random variables $\{\alpha_i\}$ and $\{\eij\}$ are independent and identically distributed and 	there is a  $\delta>0$ such that $\E| \alpha_i|^{4+\delta}<\infty$ and $\E |e_{ij}|^{4+\delta}<\infty$
		for all $j=1,\ldots, {N_i}$, $i=1,\ldots,g$.
		
		\item  The limits $ \ssb_1 = \lim_{g \rightarrow \infty} g^{-1}\sum_{i=1}^g \sx_{i}^{(b)}$, $\bb_2 = \lim_{g \rightarrow \infty} g^{-1}\sum_{i=1}^g \sx_{i}^{(b)}\sx_{i}^{(b)^T}$ and \\$ \bb_3 = \lim_{g\to\infty}\lim_{s_{L}\to\infty} n^{-1}\sum_{i=1}^g\sum_{j\in s_i} (\sx_{ij}^{(w)}-\bar{\sx}_i^{(w)})(\sx_{ij}^{(w)}-\bar{\sx}_i^{(w)})^T$ exist, and the matrices $\bb_2$ and $\bb_3$ are positive definite. Further,  $\lim_{g\to\infty}g^{-1}\sumig|\bar{\sx}_i^{(w)}|^2<\infty$, and there exists a $\delta >0$ such that $\lim_{g\to\infty}g^{-1}\sumig|\sx_i^{(b)}|^{2+\delta}<\infty$ and $\lim_{g\to\infty}\lim_{s_{L}\to\infty}n^{-1}\sumig\sum_{j\in s_i}|\sx_{ij}^{(w)}-\bar{\sx}_i^{(w)}|^{2+\delta}  <\infty$.
		
	\end{enumerate}
	As noted in \cite{lyu2022increasing}, these are mild conditions.   Conditions A3 and A4 ensure that limits needed to ensure the existence of the asymptotic variance of the estimating function exist, and that we can establish a Lyapounov condition and hence a central limit theorem for the estimating function.  Condition A4 ensures the matrix
	\[
	\bb = \mbox{block diag}[\bb_u/\dsiga^2, 1/(2\dsiga^4), \bb_3/\dsige^2, 1/(2\dsige^4)], \qquad \mbox{ with }  \bb_u =\left[ \begin{matrix} 1 & \ssb_1^T\\ \ssb_1& \bb_2  \end{matrix}\right] ,
	\]
	is positive definite.  For later, note that $\bb_{u}=\lim_{g\to\infty}g^{-1}\sum_{i=1}^g \su_i\su_i^T$.   \cite{lyu2022increasing} established the following theorem for the normal maximum likelihood and normal REML estimators.
		\begin{theoremm}{(Lyu and Welsh, 2022)}  \label{thm1}
		Suppose  Condition A holds. Then, as $g, n_{L} \to\infty$, there are solutions $\hbomega$ to the normal maximum likelihood estimating equations and $\hbomega_R$ to the normal REML estimating equations satisfying $|\bk^{1/2}(\hbomega-\hbomega_R)|=O_p(1)$, where $\bk = \diag(g\bone_{p_b+2}^T,n\bone_{p_w+1}^T)$.  Moreover, for both estimators, we have $\hat\bxi-\dot\bxi=\bb_u^{-1}g^{-1}\sumig\su_i\alpha_i+o_p(g^{-1/2})$,  $\hat\sigma_{\alpha}^2-\dsiga^2=g^{-1}\sumig(\alpha_i^2-\dsiga^2)+o_p(g^{-1/2})$,  $\hat\bbeta_2-\dot\bbeta_2=	 \bb_3^{-1}n^{-1}\sumig\sum_{j\in s_i}(\sx_{ij}^{(w)} - \bar{\sx}_i^{(w)})\eij +o_p(n^{-1/2})$ and $\hat\sigma_e^2-\dot\sigma_e^2= n^{-1}\sumig\sum_{j\in s_i}(\eij^2-\dsige^2)+o_p(n^{-1/2})$, as $g, n_L \rightarrow \infty$,
so
		\begin{equation*}
			\bk^{1/2}(\hbomega-\dot\bomega) \xrightarrow{D} N(\bzero,\bc),
		\end{equation*}
		where  
			\begin{displaymath}
		\begin{split}
			\bc &=  \left[ \begin{matrix}
				\dsiga^2/(1- \ssb_1^T\bb_2^{-1}\ssb_1 )& - \dsiga^2 \ssb_1^T\bb_2^{-1}/(1- \ssb_1^T\bb_2^{-1}\ssb_1 ) &\E\alpha_1^3&\bzero_{[1:p_w]}&0\\
				- \dsiga^2 \bb_2^{-1}\ssb_1/(1- \ssb_1^T\bb_2^{-1}\ssb_1 ) & \dsiga^2\{\bb_2^{-1} + \bb_2^{-1}\ssb_1\ssb_1^T\bb_2^{-1}/(1- \ssb_1^T\bb_2^{-1}\ssb_1)\}&\bzero_{[p_b:1]}&\bzero_{[p_b:p_w]}&\bzero_{[p_b:1]}\\
				\E\alpha_1^3 &\bzero_{[1:p_b]}&\E\alpha_1^4-\dsiga^4&\bzero_{[1:p_w]}&0\\
				\bzero_{[p_w:1]} &\bzero_{[p_w:p_b]}&\bzero_{[p_w:1]}&\dsige^2 \bb_3^{-1}&\bzero_{[p_w:1]}\\
				0&\bzero_{[1:p_b]}&0&0& \E e_{ij}^4-\dsige^4
			\end{matrix}\right]
		\end{split}
	\end{displaymath}
and $\bzero_{[p:q]}$ denotes the $p \times q$-matrix of zeros.
	\end{theoremm}

	We use Theorem \ref{thm1} to derive the asymptotic distribution of the mixed model estimators (\ref{composite}) and (\ref{syn}). 
	We begin by establishing asymptotic linearity results for the estimators.  Recall that $k_i=(N_i-n_i)/N_i$.
	\begin{theoremm}
		Suppose Condition A holds.  Then, as $g,n_L\to\infty$, we have 
		\begin{displaymath}
			\begin{split}
				&	 \hycom-\bar{y}_{i} =	k_i\{ \bar{e}_{i(s)}-\bar{e}_{i(r)}\}+k_iO_p(n_L^{-1}+g^{-1/2}n_L^{-1/2}), \\& 
				\hysyn-\bar y_{i}	=		k_i\{ \bar{e}_{i(s)}-\bar{e}_{i(r)}\}+O_p(n_L^{-1}+g^{-1/2}n_L^{-1/2}) \quad\text{and}\\&
		\hysyn-\dot{\eta}_i = \bar{e}_{i(s)} +O_p(n_L^{-1}+g^{-1/2}n_L^{-1/2}).		
			\end{split}
		\end{displaymath}
	\end{theoremm}
	
	\begin{proof}
		From (\ref{composite}) and  (\ref{p-target1}), we can write 
		\begin{displaymath}
			\hycom-\bar y_{i}=k_i\left\lbrace \su_{i}^T(\hat \bxi-\dot\bxi)+\xbarirwT(\hat \bbeta_2-\dot\bbeta_{2})+\hat \alpha_i-\alpha_i+\bar{e}_{i(r)}\right\rbrace .
		\end{displaymath}
Using the approximation
		\begin{equation}\label{eblup approx}
			\hat{\alpha}_{i}= \alpha_i +\bar{e}_{i(s)}-  \su_i^{T}(\hat{\bxi}-\dot\bxi)+O_p(n_L^{-1}+g^{-1/2}n_L^{-1/2})
		\end{equation}	 
obtained by \cite{lyu2021asymptotics} in the proof of their Theorem 2 (see also their Supplementary Material page 14), and the result from Theorem 1 that  $\hat \bbeta_{2}-\dot \bbeta_{2}=O_p(n^{-1/2}) = O_p(g^{-1/2}n_L^{-1/2})$, as $g, n_L\to \infty$, we obtain the approximation
		\begin{equation*}\label{SAPMequation}
			\begin{split}
				\hycom-\bar y_{i}&=k_i\left\lbrace \bar{e}_{i(s)}-\bar{e}_{i(r)}+\xbarirwT(\hat\bbeta_{2}-\dot \bbeta_{2})+O_p(n_L^{-1}+g^{-1/2}n_L^{-1/2})\right\rbrace \\&	
				=	k_i\{ \bar{e}_{i(s)}-\bar{e}_{i(r)}\}+ k_iO_p(n_L^{-1}+g^{-1/2}n_L^{-1/2}).
			\end{split}
		\end{equation*}
		Similarly, from (\ref{syn}) and (\ref{p-target1}), we have
		\begin{displaymath}
			\hysyn-\bar y_{i} =\su_{i}^T(\hat\bxi-\dot\bxi)+\bar{\sx}^{(w)T}_i(\hat\bbeta_{2}-\dot\bbeta_{2})+(\hat\alpha_i-\alpha_i)-\bar{e}_{i}.
		\end{displaymath}
Substituting the approximation (\ref{eblup approx}) for $\hat{\alpha}_i$, we then obtain 
		\begin{equation} \label{useful}
			\begin{split}
				\hysyn-\bar y_{i} &=\bar{e}_{i(s)}-\bar{e}_{i}+\xbarirwT(\hat\bbeta_{2}-\dot \bbeta_{2})+O_p(n_L^{-1}+g^{-1/2}n_L^{-1/2})\\&
				=\bar{e}_{i(s)}-\bar{e}_{i}+O_p(n_L^{-1}+g^{-1/2}n_L^{-1/2})\\&
				=	k_i\{ \bar{e}_{i(s)}-\bar{e}_{i(r)}\}+O_p(n_L^{-1}+g^{-1/2}n_L^{-1/2}) , 
			\end{split}
		\end{equation}
		because $k_i\bar{e}_{i(s)}-k_i\bar{e}_{i(r)}=k_i\bar{e}_{i(s)}-\{\bar{e}_{i}-(1-k_i)\bar{e}_{i(s)}\}=\bar{e}_{i(s)}-\bar{e}_{i}$.
		
The final result follows from the fact that $\bar{y}_i = \dot\eta_i + \bar{e}_i$, so we have 
\begin{displaymath}
	\begin{split}
	& \hysyn-\dot{\eta}_i = \hysyn-(\bar y_{i} -\bar{e}_i) = \bar{e}_{i(s)} +O_p(n_L^{-1}+g^{-1/2}n_L^{-1/2}),
	\end{split}
\end{displaymath}
using (\ref{useful}).	 \end{proof} 
		We can also consider using the estimator of the small area mean to estimate $\dot{\eta}_i$.  We obtain the same leading term as for the estimator of the conditional linear predictor because $k_i\{ \bar{e}_{i(s)}-\bar{e}_{i(r)}\} + \bar{e}_i =  \bar{e}_{i(s)} - (n_i/N_i)\bar{e}_{i(s)}- k_i\bar{e}_{i(r)}+\bar{e}_{i}= \bar{e}_{i(s)}$.  However, it would be unusual in practice to use the estimator of the small area mean for this purpose,  so we only state the formal result for the estimator of the conditional linear predictor.  The estimators of the small area mean  and the conditional linear predictor have the same leading terms (so are asymptotically equivalent to first order.  However the remainders for the two estimators are different, showing that there can be higher order differences between them.   In particular, If we sample the entire $i$th cluster (so $n_i=N_i$ and $k_i=0$), we have $ \hycom-\bar{y}_{i}=0$ but  $\hysyn-\bar{y}_{i}= O_p(n_L^{-1}+g^{-1/2}n_L^{-1/2})$, so $\hysyn-\bar{y}_{i}$ is only asymptotically zero. 
		
	The asymptotic distribution of  the estimators is given by the following theorem.
	\begin{theoremm}\label{corol1}
		Suppose Condition A holds, $n_i^{1/2}/n_L \rightarrow 0$ and $n_i/N_i \to f_i$ for $0 \le f_i < 1$.   Then as $g, n_L\to\infty$, we have 
		\begin{align*}
			%
			&n_i^{1/2}(\hycom-\bar y_{i})\xrightarrow{D} N(0,(1-f_i)\dsige^2), \quad n_i^{1/2}( \hysyn-\bar y_{i})\xrightarrow{D} N(0,(1-f_i)\dsige^2), \\
			& \mbox{ and } \qquad n_i^{1/2}( \hysyn-\dot{\eta}_i)\xrightarrow{D} N(0,\dsige^2).
		\end{align*}
	\end{theoremm}
	\begin{proof} Since the $e_{ij}$ are independent, $\bar{e}_{i(s)}$ and $\bar{e}_{i(r)}$ are independent, and the central theorem ensures that $n_i^{1/2}\bar{e}_{i(s)}\xrightarrow{D} N(0,\dsige^2)$ and $(N_i-n_i)^{1/2}\bar{e}_{i(r)}\xrightarrow{D} N(0,\dsige^2)$.  Provided $f_i < 1$, it follows that
\begin{align*}
n_i^{1/2}\bar{e}_{i(r)} = k_i^{-1/2}(1-k_i)^{1/2} (N_i-n_i)^{1/2}\bar{e}_{i(r)}\xrightarrow{D} N(0,(1-f_i)^{-1}f_i\dsige^2).
\end{align*}
Consequently,   
\begin{align*}
k_in_i^{1/2}(\bar{e}_{i(s)} - \bar{e}_{i(r)}) \xrightarrow{D} N(0,(1-f_i)^2\{(1-f_i)^{-1}f_i + 1\}\dsige^2) = N(0,(1-f_i)\dsige^2),
\end{align*}	
and the theorem follows from Theorem 2. \end{proof}

\noindent Therem \ref{corol1} allows us to reach several interesting conclusions.  

\begin{itemize}
\item[i)] The asymptotic distribution of both estimators is the distribution of the target characteristic of interest $F_{\bar{y}_i}$ or $F_{\dot\eta_i}$; Therem \ref{corol1} suggests that a better approximation is the distribution of  $k_i(\bar{e}_{i(s)}-\bar{e}_{i(r)})$ plus $\bar{y}_i$ when the target is $\bar{y}_i$, or the distribution of $\bar{e}_{i(s)}$ plus $\dot{\eta}_i$  when the target is $\dot{\eta}_i$.  The first pair of random variables are uncorrelated while the second are independent (so the distribution is the convolution of the $N(0, n_i^{-1}\dot{\sigma}_e^2)$ distribution with $F_{\dot\eta_i}$).  These distributions are not the same in general, but when the random effects and error in the model are all normally distributed, these approximations are the same, being
$N( \su_i^T\dot{\bxi}+\xbariwT\dot{\bbeta}_{2},\,\dot{\sigma}_{\alpha}^2 + N_i^{-1}\dsige^2 + (1-f_i)n_i^{-1}\dsige^2) = N( \su_i^T\dot{\bxi}+\xbariwT\dot{\bbeta}_{2},\,\dot{\sigma}_{\alpha}^2 + n_i^{-1}\dsige^2)$
and
$N( \su_i^T\dot{\bxi}+\xbariwT\dot{\bbeta}_{2},\,\dot{\sigma}_{\alpha}^2 + n_i^{-1}\dsige^2)$, respectively.

\item[ii)]  The asymptotic mean squared error of $\hysyn$ for estimating $\dot{\eta}_i$ is  $n_i^{-1}\dot\sigma_{e}^2$  which is greater than or equal to $n_i^{-1}(1-f_i)\dot\sigma_{e}^2$,  the asymptotic mean squared error of $\hysyn$ (or $\hycom$) for estimating $\bar{y}_i$.  Thus, using the asymptotic mean squared error of $\hysyn$ for estimating $\dot{\eta}_i$ when we are estimating $\bar{y}_i$ is conservative.  This gives an intuitive explanation for why we might expect the Prasad-Rao estimator of the mean squared error for estimating $\dot{\eta}_i$ \citep[ see the Appendix]{prasad1990estimation} to tend to be greater than our mean squared error estimator.

\item[iii)] We can also use Theorem 1 to describe the rate at which we can estimate the asymptotic mean squared errors.  For example, we have $	\hmse_{\text{LW},i}- \mse_{\text{LW},i} =n_i^{-1}k_i(\hsige^2-\dsige^2),$
so it follows from Theorem 1 that $			k_i^{-1/2}n_i^{1/2}	n^{1/2}\left(	\hmse_{\text{LW},i}-\mse_{\text{LW},i}\right)=n^{1/2}(\hsige^2-\dsige^2)\xrightarrow{D}N(0,\E e_{11}^4-\dsige^4).$

\item[iv)] Theorem \ref{corol1} establishes that the prediction intervals (\ref{predint1}) and (\ref{predint2}) have the correct asymptotic level. 

\item[v)] The result for estimating $\dot{\eta}_i$ also holds when $f_i=1$.  In this case, for estimating $\bar{y}_i$, provided $N_i-n_i  \rightarrow \infty $, we have
$(N_i-n_i)^{1/2}(\hycom-\bar y_{i})\xrightarrow{P}0$ and $ (N_i-n_i)^{1/2}(  \hysyn-\bar y_{i})\xrightarrow{P}0.$
\end{itemize}

\section{Simulation study}\label{sec:sims}

\newcommand{\rmse}{\operatorname{RMSE}}

	In this section we present results from a model-based simulation study to compare the performance of the prediction intervals (\ref{predint1}) and (\ref{predint2}) with the Prasad-Rao interval (see Supplementary) based on the mixed model estimator of the small area mean (Sam), the mixed model estimator of the conditional linear predictor (Clp) and the proposed (LW) and Prasad-Rao (PR) estimators of their (prediction) root mean squared errors (RMSEs) discussed in the preceding sections. 
	
We generated population data with $g\in\{15,30,50\}$ small areas and $N_i$ units in each small area by making area 1 the smallest area ($N_1=N_L=40$) and then setting the remaining $N_i$ equal to the integer parts of $g-1$ independent uniform $[40,400]$ random variables.  The $N_i$ were generated once for each simulation setting so that each setting involved populations of fixed $N_i$ and hence fixed size $N= \sum_{i=1}^gN_i$.  For each setting, we generated $N$ population values of an auxiliary variable $x_{ij}$ with a cluster structure by setting $x_{ij}=3+2u_i+ 4v_{ij}$, where $u_i$ and $v_{ij}$ are independent standard normal random variables.  We centered the $x_{ij}$ about their small area means $\bar{x}_i$ to obtain the within small area variable $x_{ij}-\bar{x}_i$ and also included $\bar{x}_i$ as a between small area variable.  The $N$ population values for $y$ were generated from the model
\begin{equation} \label{simmod}
y_{ij}=\beta_{0}+\beta_1\bar{x}_{i}+\beta_2(x_{ij}-\bar{x}_{i})+\alpha_i +e_{ij},\quad j=1,\ldots, N_i,\, i=1,\ldots,g,
\end{equation}
where $\{\alpha_i\}$ were generated independently from $F_{\alpha}$ with $\E(\alpha_i)=0$ and $\var(\alpha_i)=\sigma_{\alpha}^2$, and independently, $\{e_{ij}\}$ were generated independently from $F_e$ with $\E(e_{ij})=0$ and $\var(e_{ij})=\sigma_{e}^2$. We set the true regression parameters $\dot\bbeta=[5,7,3]^T$, the variance components $\dsiga^2\in\{4,64\}$ and $\dsige^2\in \{25,100\}$, and the distributions
$F_{\alpha} = N(0, \dot{\sigma}_{\alpha}^2)$ or $F_{\alpha} = 0.3 N(0.5, 1) + 0.7 N\big(\mu,\{\dot{\sigma}_{\alpha}^2-0.375-0.7\mu^2\}/0.7\big)$, and 
$F_{e} = N(0, \dot{\sigma}_{e}^2)$ or $F_{e} = 0.3 N(0.5, 1) + 0.7 N(\dot\mu, \{\dot{\sigma}_{e}^2-0.375-0.7\dot\mu^2\big\}/0.7)$ with  $\dot\mu=-0.3\times0.5/0.7$.  The $3$ values of $g$ and $2$ for each of $\dsiga^2$, $\dsige^2$, $F_{\alpha}$ and $F_e$ produced $48$ different simulation settings.
	
	For each of the $48$ simulation settings, we generated $1000$ populations and then selected one sample via simple random sampling without replacement from each population.  We set the sample size in each small area to $n_i=25$ if $N_i<50$, $n_i=\lfloor0.5* N_i\rfloor$ if $50\le N_i\le 100$, and $n_i=\lfloor0.25* N_i\rfloor$ if $N_i>100$,  where $\lfloor  \quad\rfloor$ is the integer part function, and selected the units in each area independently by simple random sampling without replacement.  For each sample, we fitted the model (\ref{simmod}) using REML in \texttt{lmer} and computed the mixed model estimators of the small area means (Sam) and conditional linear predictors (Clp) and the root mean squared error estimates (LW and PR).  We also computed the $95\%$ prediction intervals based on these estimates described in (\ref{predint1}) and (\ref{predint2}).  We compared the performance of the intervals by examining the empirical coverage (Cvge) and the relative expected length (RLen) of the prediction intervals. 
	
For each simulation setting, we report the size $N_i$ and the sample size $n_i$ for each small area. 
For the three intervals, for every small area we computed the empirical coverage probabilities (Cvge) and the average relative lengths (Rlen) defined as
\[
\mbox{Rlen} = |\mbox{Alen} - \mbox{RMSE}_{S}|/\mbox{RMSE}_{S},
\]
where Alen is the average of the estimates of the root mean squared errors of the estimators and RMSE$_{S}$ is the square root of the average of the squared differences between the estimator and the true value over the $1000$ simulations.  (We suppress the subscript $i$ for simplicity.)  The measure Alen is proportional to the average length of the intervals and RLen measures how close the lengths of the intervals are to what they should be.  Note that ALen is the same for Sam-LW and Clp-LW because they have the same root mean squared error estimators, but RLen is different for these two intervals because the different point estimators have different RMSE$_{S}$.  On the other hand, Rlen is proportional to Alen for Clp-LW  and Clp-PR because these both have the same RMSE$_{S}$.  This means that we can evaluate the difference between LW and PR by comparing Clp-LW and Clp-PR; we can also gain insight into the difference between using Sam and Clp by comparing Sam-LW with Clp-LW.  

The full set of results is given in the Supplementary Material; we present  and discuss illustrative cases below.   Table \ref{tab1} shows the empirical coverage and the relative length of the three intervals for the setting with variances $\dsiga^2=4$ and $\dsige^2=100$ (weak within cluster correlation $0.04$) when $\alpha_{i}$ and $e_{ij}$ have normal distributions; Table \ref{tab2} shows the results for the setting with variances $\dsiga^2=64$ and $\dsige^2=100$ (moderate within cluster correlation $0.39$) when $\alpha_{i}$ has a mixture distribution and $e_{ij}$ has a normal distribution; results for the remaining 48 simulation settings are similar.  The areas are presented and labeled in order of increasing size.  Simulation standard errors for the coverage probabilities can  be obtained as $\{\text{Cvge}(1-\text{Cvge})/1000\}^{1/2}$; they are approximately 0.008 or smaller.

\newcommand{\cmnt}[1]{\ignorespaces}
\def \hfillx {\hspace*{-\textwidth} \hfill}

\begin{table}[h]			
\captionsetup{width=1.\textwidth }
\centering
\caption{Simulated coverage and length of prediction intervals when $\alpha_{i}$ and $e_{ij}$ have  normal distributions with variances $\dsiga^2=4$ and $\dsige^2=100$, respectively.}\label{tab1}
\begin{tabular}{cccccccccccc}
\toprule
\multicolumn{3}{c}{}                                  &  & \multicolumn{2}{c}{Sam-LW} &  & \multicolumn{2}{c}{Clp-LW} &  & \multicolumn{2}{c}{Clp-PR} \\\cline{1-3}\cline{5-6} \cline{8-9}\cline{11-12}
Area & N\_i & n\_i & 
& Cvge         & Rlen        &  & Cvge         & Rlen        &  & Cvge         & Rlen        \\\hline
1    & 40   & 20   \cmnt{& -0.03 }                                 &  & 0.965        & 0.10         &  & 0.925        & 0.09        &  & 0.941        & 0.18        \\
2    & 51   & 26   \cmnt{& -0.01}                                  &  & 0.965        & 0.08        &  & 0.966        & 0.05        &  & 0.976        & 0.30         \\
3    & 82   & 41   \cmnt{& 0.01}                                   &  & 0.962        & 0.08        &  & 0.934        & 0.05        &  & 0.987        & 0.46        \\
4   & 86   & 43   \cmnt{& 0.05}                                   &  & 0.980         & 0.19        &  & 0.975        & 0.13        &  & 0.976        & 0.33        \\
5    & 100  & 25   \cmnt{& 0.02}                                   &  & 0.971        & 0.17        &  & 0.965        & 0.10         &  & 0.973        & 0.33        \\
6    & 110  & 28   \cmnt{& -0.03}                                  &  & 0.976        & 0.18        &  & 0.967        & 0.13        &  & 0.961        & 0.27        \\
7    & 113  & 28   \cmnt{& -0.03}                                  &  & 0.972        & 0.17        &  & 0.97         & 0.13        &  & 0.978        & 0.30         \\
8    & 120  & 30   \cmnt{& 0.02}                                   &  & 0.967        & 0.15        &  & 0.966        & 0.11        &  & 0.967        & 0.29        \\
9   & 122  & 30   \cmnt{& 0.05}                                   &  & 0.971        & 0.09        &  & 0.933        & 0.06        &  & 0.969        & 0.32        \\
10    & 135  & 34   \cmnt{& 0.03}                                   &  & 0.962        & 0.07        &  & 0.931        & 0.07        &  & 0.993        & 0.45        \\
11   & 141  & 35   \cmnt{& -0.03}                                  &  & 0.973        & 0.15        &  & 0.966        & 0.11        &  & 0.975        & 0.36        \\
12   & 147  & 37   \cmnt{& 0}                                      &  & 0.963        & 0.10         &  & 0.962        & 0.06        &  & 0.988        & 0.35        \\
13    & 150  & 38   \cmnt{& 0.01}                                   &  & 0.967        & 0.11        &  & 0.962        & 0.06        &  & 0.968        & 0.24        \\
14   & 152  & 38   \cmnt{& -0.04}                                  &  & 0.976        & 0.11        &  & 0.968        & 0.06        &  & 0.979        & 0.29        \\
15   & 175  & 44   \cmnt{& -0.02}                                  &  & 0.967        & 0.13        &  & 0.96         & 0.10         &  & 0.975        & 0.36         \\\bottomrule
\end{tabular}
\end{table}
\begin{table}
\captionsetup{width=1.\textwidth }
\centering
\caption{Simulated coverage and length of prediction intervals when $\alpha_{i}$ has a mixture distribution and $e_{ij}$ has a  normal distribution with variances $\dsiga^2=64$ and $\dsige^2=100$, respectively.}\label{tab2}
\begin{tabular}{cccccccccccc}
\toprule
\multicolumn{3}{c}{Method}                                  &  & \multicolumn{2}{c}{Sam-LW} &  & \multicolumn{2}{c}{Clp-LW} &  & \multicolumn{2}{c}{Clp-PR} \\\cline{1-3}\cline{5-6} \cline{8-9}\cline{11-12}
Area & N\_i & n\_i \cmnt{& $\hat{\alpha}_i/\hat\sigma_{\alpha}$} &  & Cvge         & Rlen        &  & Cvge         & Rlen        &  & Cvge         & Rlen        \\\hline
1    & 40   & 20   \cmnt{& -0.04}                                  &  & 0.957 & 0.05 &  & 0.934 & 0.03 &  & 0.985 & 0.50  \\
2   & 45   & 20   \cmnt{& -0.02}                                 &  & 0.950  & 0.03 &  & 0.937 & 0.03 &  & 0.987 & 0.42 \\
3    & 58   & 29   \cmnt{& -0.01}                                  &  & 0.957 & 0.01 &  & 0.945 & 0.06 &  & 0.985 & 0.47 \\
4    & 65   & 32   \cmnt{& 0.04}                                   &  & 0.965 & 0.02 &  & 0.945 & 0.03 &  & 0.993 & 0.51 \\
5    & 81   & 40   \cmnt{& 0.01}                                   &  & 0.950  & 0.01 &  & 0.936 & 0.06 &  & 0.989 & 0.46 \\
6    & 85   & 42   \cmnt{& -0.05}                                  &  & 0.957 & 0.01 &  & 0.947 & 0.02 &  & 0.995 & 0.52 \\
7   & 103  & 26   \cmnt{& -0.03}                                &  & 0.960  & 0.05 &  & 0.949 & 0.03 &  & 0.984 & 0.31\\
8   & 109  & 27   \cmnt{& 0.07 }                                &  & 0.965 & 0.06 &  & 0.960  & 0.04 &  & 0.986 & 0.31 \\
9    & 115  & 29   \cmnt{& 0.05}                                   &  & 0.957 & 0.06 &  & 0.952 & 0.04 &  & 0.991 & 0.33 \\
10   & 150  & 38   \cmnt{& -0.06}                                 &  & 0.963 & 0.08 &  & 0.959 & 0.07 &  & 0.984 & 0.36 \\
11   & 151  & 38   \cmnt{& 0.05}                                   &  & 0.946 & 0.01 &  & 0.945 & 0.01 &  & 0.977 & 0.28 \\
12    & 162  & 40   \cmnt{& -0.04 }                              &  & 0.952 & 0.06 &  & 0.951 & 0.05 &  & 0.986 & 0.33 \\
13    & 163  & 41   \cmnt{& -0.03}                               &  & 0.946 & 0.02 &  & 0.947 & 0.01 &  & 0.988 & 0.28 \\
14    & 180  & 45  \cmnt{& 0.03}                                &  & 0.947 & 0.00  &  & 0.947 & 0.01 &  & 0.985 & 0.25 \\
15   & 193  & 48   \cmnt{& 0.02}                                 &  & 0.961 & 0.06 &  & 0.958 & 0.06 &  & 0.989 & 0.34 \\
\bottomrule
\end{tabular}
\end{table}

The simulation results show that our asymptotic results based on both $g$ and $n_L$ going to infinity provide useful approximations that work well even when $g=15$ and $n_L=20$.   The empirical coverages of all three intervals are close to the nominal level and tend to the nominal level as $g$ and $n_L$ increase, confirming our large $g$ and $n_L$ asymptotic results.  Comparing the empirical coverage and relative length for Clp-LW with Clp-PR, we see that the \cite{prasad1990estimation} mean squared error estimator PR is typically larger than our proposed estimator LW, so Clp-PR is more conservative (with wider intervals) than Clp-LW.  The differences in performance between Sam-LW and Clp-LW are smaller with both performing well.  


\section{Consumer expenditure on fresh milk products}\label{sec:example}

We obtained data from the Dairy Survey component of the 2002 Consumer Expenditure Survey conducted by the U.S. Bureau of the Census for the U.S. Bureau of Labor Statistics; the data are available from \url{https://www.bls.gov/cex/pumd_data.htm}.  We treated the consumer expenditure on fresh milk products (MILKPROD) as the survey variable of interest and considered the problem of estimating the average consumer expenditure on fresh milk products in different states (small areas).  We used the total expenditure on food (FOODTOT), the number of persons under age 18 in the family (PERSLT18) and the total family income before taxes in the last 12 months (FINCBEFX) as the auxiliary variables.  This data set is similar to that used in \cite{arora1997superiority}; they used the earlier 1989 survey, focussed on the expenditure on fresh whole milk, and potentially used different auxiliary variables.  If we knew the means of the auxiliary variables for each state, we could use our methods to estimate the average expenditure on fresh milk products in 2002 in each state.  As we do not have this information, we instead treated the data set as a pseudo-population and sampled from it.  We repeated this sampling $1000$ times, effectively implementing a design-based simulation from our fixed population to evaluate the design-based properties of our proposed model-based methods.  

In creating the population, we discarded 6 states with fewer than 10 observations, leaving us with $N=4022$ observations from $g=34$ states with between $N_L=36$ and $N_U=397$ observations from each state.  We centered the auxiliary variables about their cluster means (adding $cent$ to their variable name) and then included the cluster means (adding $avg$ to the variable name) as between state variables so that we have $p_b=3$ plus $p_w=3$ auxiliary variables.  The cluster means give the average per family of each variable for each state.  
We selected the $1000$ samples independently by simple random sampling without replacement from each state with $n_L = 20$ by setting $n_i=20$ if $N_i<50$, $n_i=\lfloor0.5* N_i\rfloor$ if $50\le N_i\le 100$, and $n_i=\lfloor0.25* N_i\rfloor$ if $N_i>100$,  where $\lfloor  \quad\rfloor$ is the integer part function.  In each sample, we fitted the model (\ref{milkmodel}) 
\begin{align}  \label{milkmodel}
MILKPROD_{ij}=&\beta_{0}+\beta_1 FOODTOTavg_i+\beta_2 PERSLT18avg_i+\beta_3FINCBEFXavg_i\\&
	+\beta_4 FOODTOTcent_{ij}+\beta_5 PERSLT18cent_{ij}+\beta_6 FINCBEFXcent_{ij}+\alpha_{i}+e_{ij},\nonumber
\end{align}
using \texttt{lmer} and then used the parameter estimates to compute the mixed model estimators of the small area means (Sam) and the conditional linear predictors (Clp)  and their prediction mean squared errors.  We then computed the $95\%$ model-based prediction intervals (\ref{predint1}) and (\ref{predint2}).  Similar to the model-based simulation, we examined the empirical design-coverage (Cvge) and the relative design-expected length (Rlen) of the prediction intervals; the results over $1000$ samples together with the standardised population EBLUPs $\hat\alpha_i/\hat{\sigma}_{\alpha}$ for each state are shown in Table \ref{tab6}.  The relative design-bias and the design RMSEs of Sam and Clp, together with the design-averages of the LW and PR estimators of the RMSEs are available in the Supplementary Information.



\begin{table}[!h]
\captionsetup{width=1.\textwidth }
\centering
\caption{Simulated design-coverage and design-expected length of nominal $95\%$ confidence intervals  for the average consumer expenditure on fresh milk products in each state in 2002.  $^*$ identifies states in Group 3 and $^\dagger$ identifies states in Group 2. }\label{tab6}
\begin{tabular}{ccccccccccccc}
\toprule
\multicolumn{4}{c}{$\rmse$}                                  &  & \multicolumn{2}{c}{Sam-LW} &  & \multicolumn{2}{c}{Clp-LW} &  & \multicolumn{2}{c}{Clp-PR} \\\cline{1-4} \cline{6-7}\cline{9-10}\cline{12-13}
STATE    & N\_i    & n\_i   & $\hat\alpha_i/\hat{\sigma}_{\alpha}$ &  & Cvge         & Rlen        &  & Cvge         & Rlen        &  & Cvge         & Rlen        \\\hline
16    & 36   & 20   & 0.14                                   &  & 0.997 & 0.57 &  & 1.000     & 0.94 &  & 0.996 & 1.13 \\
50$^\dagger$    & 37   & 20   & 0.52                 &  & 0.912 & 0.12 &  & 0.926 & 0.25 &  & 0.855 & 0.19 \\
31    & 42   & 20   & -0.33                                  &  & 1.000     & 0.58 &  & 1.000     & 0.28 &  & 0.981 & 0.29 \\
22$^\dagger$    & 43   & 20   & -0.48                &  & 1.000     & 0.56 &  & 1.000     & 0.08 &  & 0.917 & 0.08 \\
21    & 44   & 20   & -0.17                                  &  & 1.000     & 0.92 &  & 1.000     & 0.85 &  & 0.992 & 0.84 \\
15    & 45   & 20   & 0.10                                    &  & 1.000     & 0.85 &  & 1.000     & 1.22 &  & 0.995 & 1.18 \\
32$^\dagger$    & 46   & 20   & 0.71                 &  & 0.921 & 0.14 &  & 0.832 & 0.35 &  & 0.698 & 0.37 \\
37$^\dagger$    & 47   & 20   & -0.48                &  & 1.000     & 0.36 &  & 1.000     & 0.05 &  & 0.915 & 0.09 \\
1     & 52   & 26   & -0.01                                  &  & 1.000     & 1.54 &  & 1.000     & 2.06 &  & 1.000  & 2.53 \\
45    & 53   & 26   & 0.04                                   &  & 0.973 & 0.14 &  & 0.997 & 0.43 &  & 0.991 & 0.63 \\
2$^*$     & 58   & 29   & 1.22                                   &  & 0.795 & 0.32 &  & 0.414 & 0.56 &  & 0.571 & 0.47 \\
9$^*$     & 59   & 30   & -1.08                                  &  & 0.745 & 0.37 &  & 0.360  & 0.59 &  & 0.554 & 0.49 \\
41    & 65   & 32   & 0.12                                   &  & 0.984 & 0.25 &  & 0.997 & 0.70  &  & 0.998 & 1.10  \\
49    & 67   & 34   & 0.00                                      &  & 0.995 & 0.48 &  & 1.000     & 0.75 &  & 0.999 & 1.24 \\
18    & 71   & 36   & -0.27                                  &  & 0.971 & 0.02 &  & 0.996 & 0.25 &  & 0.996 & 0.63 \\
27$^*$    & 76   & 38   & 1.31                                   &  & 0.740  & 0.39 &  & 0.449 & 0.56 &  & 0.658 & 0.43 \\
8$^\dagger$     & 82   & 41   & 0.66                   &  & 0.867 & 0.23 &  & 0.840  & 0.28 &  & 0.931 & 0.03 \\
13    & 93   & 46   & -0.09                                  &  & 0.991 & 0.30  &  & 0.998 & 0.76 &  & 1.000     & 1.43 \\
24$^\dagger$    & 94   & 47   & -0.91                 &  & 0.933 & 0.06 &  & 0.790  & 0.35 &  & 0.908 & 0.09 \\
29    & 98   & 49   & -0.26                                  &  & 1.000     & 0.78 &  & 1.000     & 0.75 &  & 1.000 & 1.47 \\
53    & 99   & 50   & 0.24                                   &  & 0.998 & 0.63 &  & 0.999 & 0.79 &  & 1.000     & 1.56 \\
55    & 119  & 30   & -0.49                                  &  & 1.000     & 0.47 &  & 1.000  & 0.31 &  & 0.976 & 0.30 \\ 
51    & 122  & 30   & 0.01                                   &  & 0.998 & 0.80  &  & 0.998 & 1.21 &  & 1.000     & 1.19 \\
25    & 126  & 32   & 0.06                                   &  & 1.000     & 1.06 &  & 1.000  & 1.55 &  & 1.000 & 1.59 \\
4     & 133  & 33   & 0.13                                   &  & 0.999 & 0.55 &  & 1.000     & 0.79 &  & 0.997 & 0.83 \\
26    & 139  & 35   & -0.29                                  &  & 0.997 & 0.51 &  & 1.000     & 0.60  &  & 0.996 & 0.68 \\
34    & 160  & 40   & -0.16                                  &  & 0.998 & 0.61 &  & 1.000     & 0.84 &  & 1.000  & 1.00    \\
17    & 161  & 40   & -0.28                                  &  & 1.000  & 0.81 &  & 1.000  & 0.86 &  & 1.000     & 1.02 \\
39    & 229  & 57   & -0.16                                  &  & 0.999 & 0.64 &  & 0.999 & 0.87 &  & 1.000     & 1.23 \\
42    & 261  & 65   & -0.57                                  &  & 0.981 & 0.15 &  & 0.988 & 0.11 &  & 0.993 & 0.36 \\
36    & 280  & 70   & 0.02                                   &  & 0.975 & 0.16 &  & 0.984 & 0.30  &  & 0.997 & 0.63 \\
12    & 283  & 71   & 0.58                                   &  & 0.982 & 0.12 &  & 0.979 & 0.09 &  & 0.991 & 0.36 \\
48    & 305  & 76   & -0.44                                  &  & 0.986 & 0.27 &  & 0.991 & 0.30  &  & 1.000     & 0.65 \\
6     & 397  & 99   & 0.59                                   &  & 0.964 & 0.05 &  & 0.968 & 0.04 &  & 0.998 & 0.37 \\
\bottomrule  
\end{tabular}
\end{table}

To interpret the results, we partition the states into three groups: Group 3 with three states $\{2, 9, 27\}$ for which the design-coverage of all three intervals is well below the nominal level; Group 2 with six states $\{8, 22, 24, 32, 37, 50\}$ for which at least one interval has design-coverage below the nominal level, but not all intervals perform poorly; and Group 3 with the remaining twenty-five states for which the design-coverage of all three intervals is above the nominal level.  

For the states in Group 1,  the design-coverages tend to be conservative and similar across all three intervals.  The design-biases for Sam are smaller than for Clp except in state 49 where the design-biases are very similar.  The design-average LW and PR estimators of the RMSE are greater than the simulation RMSEs of the Sam and Clp estimators, respectively.  The design-average LW estimator is smaller than the design-average PR estimator in all except 4 states (15, 21, 51 and 55) where PR has a slightly smaller design-average than LW.

Group 3 comprises the three states for which all three intervals have design-coverage well below the nominal level. Sam-LW has the best design-coverage, followed by Clp-PR and then Clp-LW.   Both point estimates Sam and Clp are design-biased with Clp  having greater design-bias than Sam.  The design-average RMSE estimator PR is on design-average greater than LW, explaining why Clp-PR has better design-coverage than Clp-LW.  However, Sam-LW has better design-coverage than Clp-PR because PR is not large enough to overcome the greater design-bias in Clp than in Sam. 

Group 2 comprises the states for which at least one but not all three intervals have design-coverage below the nominal level.  Usually Clp-LW but sometimes Clp-PR performs most poorly; Sam-LW usually performs better, but in state 8 Sam-LW also has low design-coverage.  Clp has larger design-bias than Sam for Group 2.  The design-average PR estimator of the RMSE is smaller than the simulation RMSE of the Clp estimator, while the design-average LW estimators of the RMSE is close to the simulation RMSEs of the Sam estimators. The RMSE estimator LW is on design-average smaller than PR; the larger PR reduces the effect of the design-bias of Clp on the design-coverage in state 8.    

These results show that the Sam-LW intervals are generally preferable to the Clp intervals for making design-based inference in our consumer expenditure on fresh milk products population.  Nonetheless, we did not expect to see results like those in Group 3 so we felt that further exploration of why these results occur would be interesting and useful.

As we have access to the whole population (which is not usual in practice), we are able to explore the population.  We used \texttt{lmer} from the R package \texttt{lme4} to fit the nested error regression model (\ref{milkmodel})
to the population data.  Figure \ref{fig1} shows a normal QQ-plot of the EBLUPs (with approximate $95\%$ prediction intervals computed as in \cite{lyu2021asymptotics}) and a normal QQ-plot of the errors from the fitted model.   Figure \ref{fig1} shows that it is plausible to treat the random effects as approximately normally distributed; the errors have an asymmetric long-tailed distribution and it is not plausible to treat the errors in the model as normally distributed.  In a model-based analysis, rather than simply relying on the asymptotic theory, we could consider using transformations to improve the fit, make predictions on the transformed scale and then invert the prediction intervals (possibly adjusting for back-transformation bias).  However, for our design-based analysis, as is arguably usual in practice, we keep the data on the raw scale.   Table \ref{tab5} shows the parameter estimates and standard errors for the fitted model.   Although some coefficients are not significant, we retain them in the model.  We see that $\hat{\sigma}_e^2/\hat{\sigma}_a^2=48.72$ is large so that the within cluster correlation is very small (approximately $0.02$).    


\def \hfillx {\hspace*{-\textwidth} \hfill}
\begin{table}[!h]
\captionsetup{width=.5\textwidth }
\centering
\caption{Parameter estimates (REML) for modelling the consumer expenditure on fresh milk products population.} 
\begin{tabular}{lll}
\toprule
Effect& Estimate& Std Error\\\hline
(Intercept)&2.6219 & 0.7622\\
FOODTOTavg &1.1617 &5.0871\\
PERSLT18avg &0.9235 &0.5526\\
FINCBEFXavg& 0.0155 &0.0083\\
FOODTOTcent& 4.7930& 0.3656\\
PERSLT18cent&0.7098 &0.0399\\
FINCBEFXcent&0.0015 &0.0010\\\hline
$\siga^2$&0.1757&\\
$\sige^2$&8.5600&\\
\bottomrule
\multicolumn{3}{l}{NOTE: lmer does not compute standard }\\
\multicolumn{3}{l}{errors for the variance components.}
\end{tabular}
\label{tab5}
\end{table}

\begin{figure}[htbp]
\begin{center}
\includegraphics[width=0.45\linewidth, height=4cm]{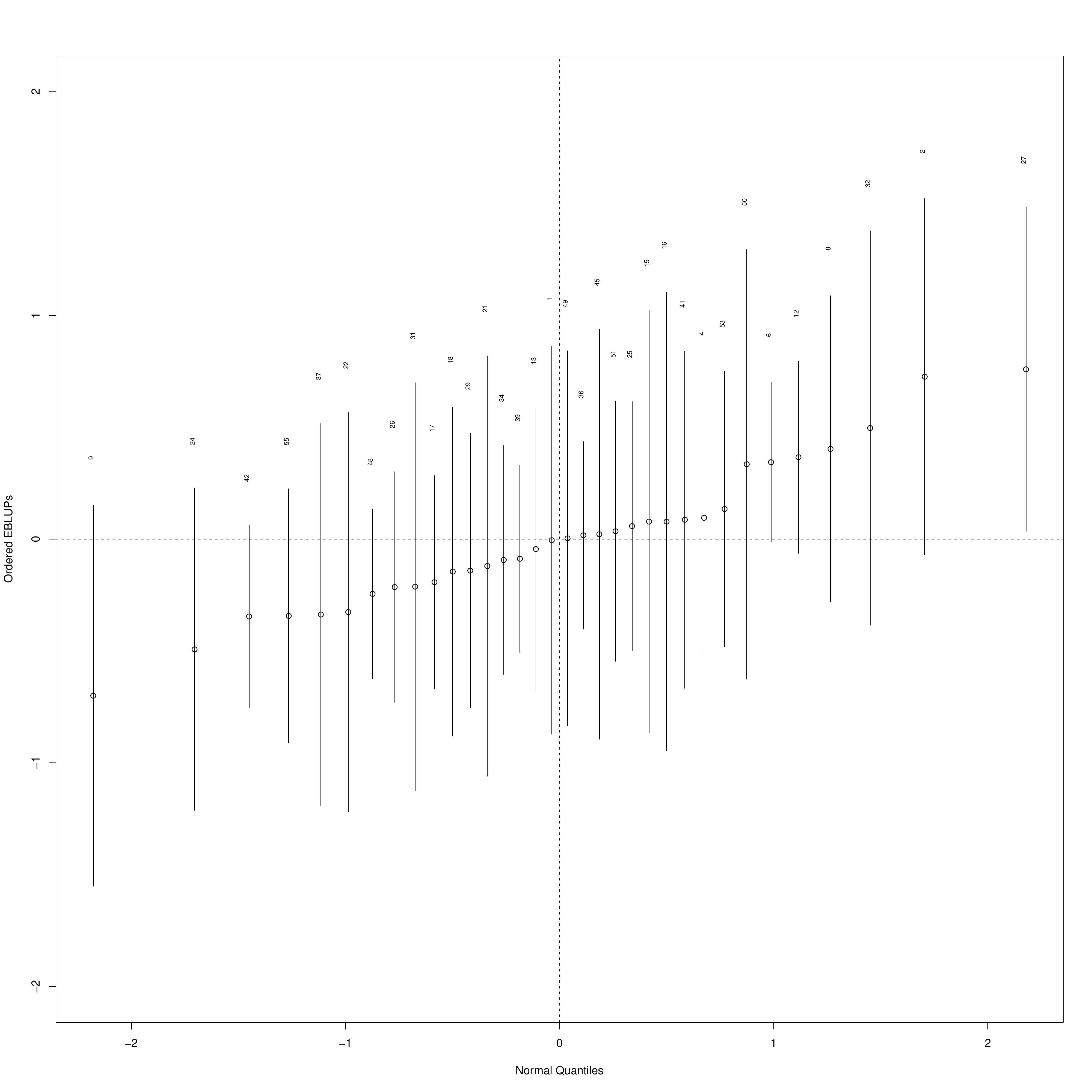}
\includegraphics[width=0.45\linewidth, height=4cm]{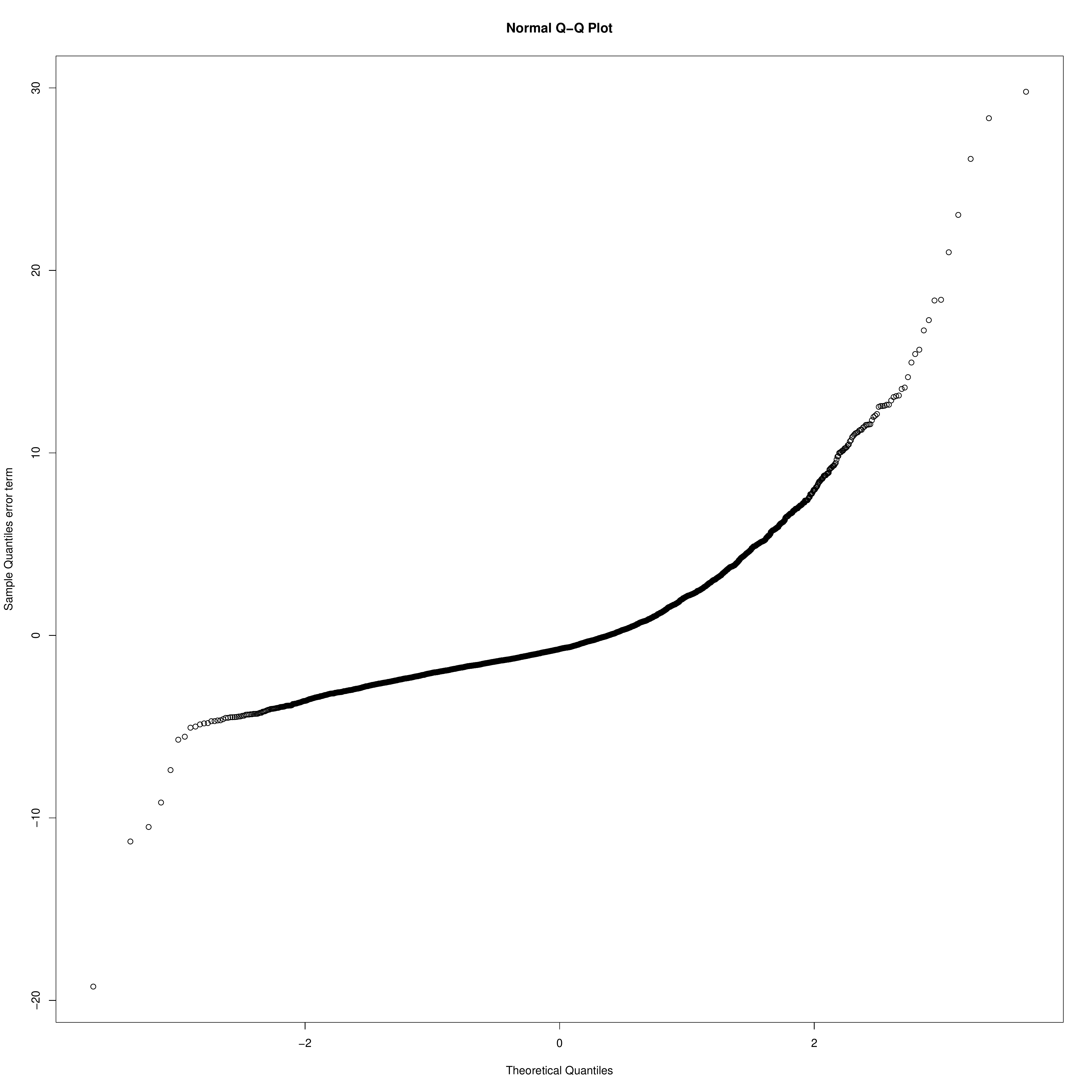}
\caption{Normal QQ-plots for the EBLUPs $\hat{\alpha}_i$ and estimated errors $\hat{e}_{ij}$ from the model (\ref{milkmodel}) fitted to the consumer expenditure on fresh milk products in 2002 data.}
\label{fig1}
\end{center}
\end{figure}


The normal QQ-plot of the EBLUPs in Figure \ref{fig1} shows that the Group 3 states have extreme EBLUPs and the Group 2 states have EBLUPs in the tails of the distribution but these are mixed in with some EBLUPs from Group 1 states.  This mixing suggests that the EBLUPS alone do not identify the group to which a state belongs.   Another potentially important value is the sample size $n_i$ in each state.  Figure \ref{fig2} shows the population standardised EBLUPS $\hat{\alpha}_i/\hat{\sigma}_{\alpha}$ plotted against the sample size $n_i$ for each state; these variables are also included in Table \ref{tab6}.  States plotted in the top and bottom left of the plot (high standardised EBLUPs and small to moderate sample size) are in Groups 2 and 3; states at the top or bottom right (small standardised EBLUPs or large standardised EBLUPS with large sample sizes) are in Group 1 with all the others states.  Specifically, states 6 ($n_6=99$), 12 ($n_{12}=71$) and 42 ($n_{42}=65$)) have relatively extreme EBLUPs but larger sample sizes so these states are in Group 1.  This suggests that both the magnitude of the standardised EBLUPs and the sample size determine the difficulty of estimating a particular state. 

\begin{figure}[htbp]
\begin{center}
\includegraphics[width=0.5\linewidth, height=5cm]{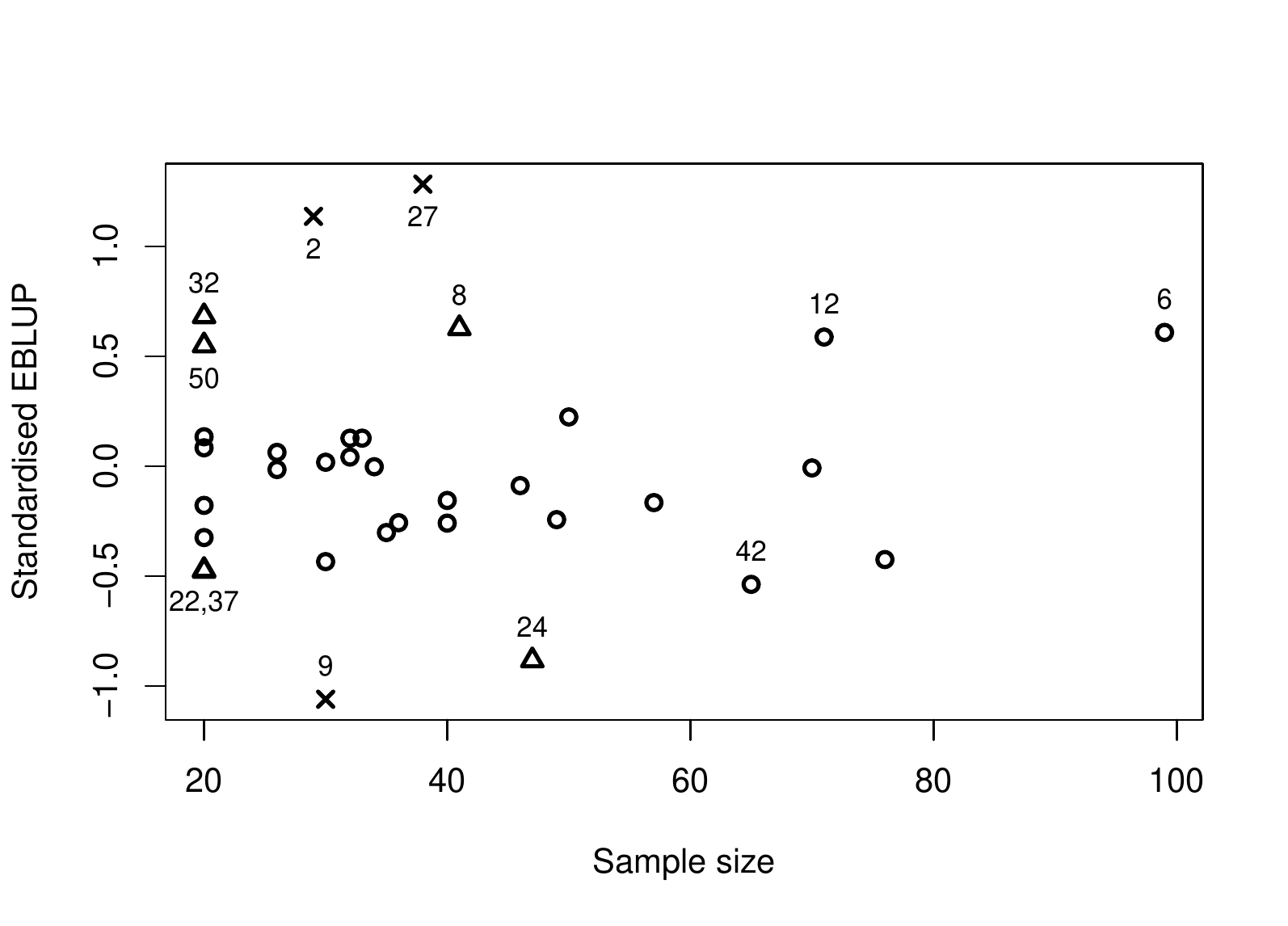}
\caption{Plot of the population standardised EBLUPS $\hat{\alpha}_i/\hat{\sigma}_{\alpha}$ against sample size in each state for the consumer expenditure on fresh milk products. Group 1 states are plotted as circles, Group 2 as triangles and Group 3 as crosses. Selected states are labelled by state number.}
\label{fig2}
\end{center}
\end{figure}

Another possibility we considered is that the within state variances may not be constant, either because of lack of fit in the mean in some states or because the within state variance simply differs in different states.  Incorrectly assuming constant variance as we have done in model (\ref{milkmodel}) would mean that the LW estimator is too small in states with large within state variance, so the intervals should perform poorly in these states.  We computed the sample variance of the errors in each state to produce rough estimates of the within state variances.  These variances range from $1.48$ to $20.46$.  In Group 3, states 2 and 27 have large  variances, but state 9 does not.  In Group 2, states 8, 32 and 50 have larger error variances (group 8 has the larges variance of all the states), but states 22, 24 and 37 actually have very small error variances.  States 18 and 36 in Group 1 have variances comparable to the larger variances in Group 2 while the majority of states have small variances.  The patterns are less consistent than those based on the standardised EBLUPs and sample size so we focussed more on these variables.

We cannot rule out the possibility that the results of the design-based simulation using the consumer expenditure on fresh milk products data are at least in part due to failures of the model (\ref{milkmodel}).  However, we can  explore whether similar results occur when the model is correct by carrying out a set of additional design-based simulations. We used the same 48 settings as in our model-based simulation, but we generated a single population for each setting and then selected $1000$ samples from it by simple random sampling without replacement.  We fitted the model and computed the same $95\%$ prediction intervals as before and evaluated their design-based properties.  The full set of results are available in the Supplementary material; we present and discuss illustrative cases below.

\begin{table}[!h]
\captionsetup{width=1\textwidth }
\centering
\caption{Simulated design-coverage and design-average length of confidence intervals when $\alpha_{i}$ and $e_{ij}$ have  normal distributions with variances $\dsiga^2=4$ and $\dsige^2=100$, respectively.  $^*$ identifies states in Group 3 and $^\dagger$ identifies states in Group 2.}\label{tab3}
\begin{tabular}{ccccccccccccc}
\toprule
\multicolumn{4}{c}{Method}                                  &  & \multicolumn{2}{c}{Sam-LW} &  & \multicolumn{2}{c}{Clp-LW} &  & \multicolumn{2}{c}{Clp-PR} \\\cline{1-4}\cline{6-7} \cline{9-10}\cline{12-13}
Area & N\_i & n\_i & $\hat{\alpha}_i/\hat\sigma_{\alpha}$ &  & Cvge         & Rlen        &  & Cvge         & Rlen        &  & Cvge         & Rlen        \\\hline
1    & 40   & 20   & 0.37                                   &  & 0.993 & 0.30  &  & 0.999 & 0.27 &  & 1.000     & 0.65 \\
2   & 40   & 20   & 0.11                                   &  & 0.997 & 0.59 &  & 1.000     & 1.25 &  & 1.000     & 1.93 \\
3   & 43   & 20   & 0.12                                   &  & 0.971 & 0.15 &  & 1.000     & 0.69 &  & 1.000     & 1.13 \\
4   & 50   & 25   & 0.48                                   &  & 0.980  & 0.21 &  & 0.991 & 0.10  &  & 0.999 & 0.53 \\
5    & 61   & 30   & 0.24                                   &  & 0.994 & 0.36 &  & 1.000     & 0.61 &  & 1.000     & 1.32 \\
6   & 72   & 36   & 0.37                                   &  & 0.981 & 0.25 &  & 0.992 & 0.34 &  & 1.000     & 1.02 \\
7   & 74   & 37   & 0.51                                   &  & 0.979 & 0.09 &  & 0.987 & 0.09 &  & 0.999 & 0.65 \\
8$^*$   & 79   & 40   & -1.45                                  &  & 0.860  & 0.25 &  & 0.554 & 0.50  &  & 0.901 & 0.22 \\
9   & 81   & 40   & -0.60                                   &  & 0.978 & 0.14 &  & 0.981 & 0.06 &  & 1.000     & 0.62 \\
10$^\dagger$   & 88   & 44   & 1.03                 &  & 0.925 & 0.07 &  & 0.847 & 0.30  &  & 0.989 & 0.10  \\
11   & 89   & 44   & 0.71                                   &  & 0.982 & 0.16 &  & 0.978 & 0.03 &  & 0.999 & 0.50  \\
12   & 92   & 46   & -0.43                                  &  & 0.992 & 0.35 &  & 0.996 & 0.32 &  & 1.000     & 1.08 \\
13   & 96   & 48   & -0.61                                  &  & 0.962 & 0.04 &  & 0.971 & 0.01 &  & 1.000     & 0.59 \\
14   & 101  & 25   & 0.14                                   &  & 0.997 & 0.49 &  & 0.999 & 0.79 &  & 1.000     & 1.03 \\
15$^*$   & 105  & 26   & 1.13                                   &  & 0.909 & 0.19 &  & 0.831 & 0.32 &  & 0.880  & 0.22 \\
16   & 113  & 28   & 0.26                                   &  & 0.992 & 0.32 &  & 0.995 & 0.41 &  & 0.999 & 0.65     \\
17   & 127  & 32   & 0.07                                   &  & 0.995 & 0.44 &  & 0.999 & 0.71 &  & 1.000     & 1.05 \\
18   & 129  & 32   & 0.54                                   &  & 0.986 & 0.23 &  & 0.989 & 0.18 &  & 0.999 & 0.42 \\
19    & 131  & 33   & -0.31                                  &  & 0.992 & 0.30  &  & 0.996 & 0.41 &  & 1.000     & 0.71 \\
20$^\dagger$   & 138  & 34   & -0.91                 &  & 0.934 & 0.08 &  & 0.914 & 0.18 &  & 0.962 & 0.01 \\
21    & 152  & 38   & 0.15                                   &  & 0.996 & 0.47 &  & 0.998 & 0.66 &  & 1.000     & 1.06 \\
22$^\dagger$   & 162  & 40   & -1.25                 &  & 0.927 & 0.16 &  & 0.854 & 0.30  &  & 0.950  & 0.12 \\
23    & 169  & 42   & -0.74                                  &  & 0.969 & 0.07 &  & 0.971 & 0.01 &  & 0.992 & 0.28 \\
24    & 170  & 42   & -0.01                                  &  & 0.987 & 0.36 &  & 0.996 & 0.55 &  & 1.000     & 0.96 \\
25    & 173  & 43   & -0.11                                  &  & 0.995 & 0.47 &  & 0.997 & 0.67 &  & 1.000     & 1.12 \\
26   & 177  & 44   & 0.78                                   &  & 0.965 & 0.05 &  & 0.966 & 0.01 &  & 0.996 & 0.28 \\
27   & 186  & 46   & -0.08                                  &  & 0.991 & 0.34 &  & 0.997 & 0.51 &  & 1.000     & 0.93 \\
28    & 190  & 48   & 0.09                                   &  & 0.997 & 0.42 &  & 0.999 & 0.59 &  & 1.000     & 1.05 \\
29    & 197  & 49   & -0.84                                  &  & 0.968 & 0.04 &  & 0.962 & 0.03 &  & 0.992 & 0.25 \\
30   & 199  & 50   & 0.23                                   &  & 0.989 & 0.25 &  & 0.991 & 0.34 &  & 1.000     & 0.74 \\
\bottomrule
\end{tabular}
\end{table}

\begin{table}
\captionsetup{width=1\textwidth }
\centering
\caption{Simulated design-coverage and design-average length of confidence intervals when $\alpha_{i}$ has a mixture distribution and $e_{ij}$ has a  normal distribution with variances $\dsiga^2=4$ and $\dsige^2=100$, respectively.  $^*$ identifies states in Group 3 and $^\dagger$ identifies states in Group 2.}\label{tab4}
\begin{tabular}{ccccccccccccc}
\toprule
\multicolumn{4}{c}{Method}                                  &  & \multicolumn{2}{c}{Sam-LW} &  & \multicolumn{2}{c}{Clp-LW} &  & \multicolumn{2}{c}{Clp-PR} \\\cline{1-4}\cline{6-7} \cline{9-10}\cline{12-13}
Area & N\_i & n\_i & $\hat{\alpha}_i/\hat\sigma_{\alpha}$ &  & Cvge         & Rlen        &  & Cvge         & Rlen        &  & Cvge         & Rlen        \\\hline
1$^*$    & 40   & 20   & -1.27                                  &  & 0.883 & 0.27 &  & 0.205 & 0.59 &  & 0.509 & 0.5  \\
2   & 44   & 20   & 0.26                                   &  & 0.997 & 0.49 &  & 1.000     & 0.65 &  & 0.999 & 0.91 \\
3   & 54   & 27   & 0.03                                   &  & 0.996 & 0.50  &  & 1.000     & 1.13 &  & 1.000     & 1.86 \\
4    & 58   & 29   & -0.24                                  &  & 0.976 & 0.16 &  & 0.998 & 0.49 &  & 1.000     & 1.04 \\
5   & 68   & 34   & 0.35                                   &  & 0.978 & 0.13 &  & 0.994 & 0.24 &  & 0.999 & 0.78 \\
6   & 80   & 40   & 0.15                                   &  & 0.987 & 0.22 &  & 0.998 & 0.56 &  & 1.000     & 1.31 \\
7   & 80   & 40   & -0.68                                  &  & 0.958 & 0.03 &  & 0.948 & 0.10  &  & 0.996 & 0.34 \\
8$^\dagger$    & 91   & 46   & -0.99                  &  & 0.971 & 0.04 &  & 0.872 & 0.28 &  & 0.987 & 0.10  \\
9$^\dagger$   & 95   & 48   & 0.91                  &  & 0.963 & 0.01 &  & 0.887 & 0.24 &  & 0.989 & 0.17 \\
10$^\dagger$   & 96   & 48   & 1.04                  &  & 0.920  & 0.14 &  & 0.794 & 0.34 &  & 0.967 & 0.01 \\
11   & 101  & 25   & 0.34                                   &  & 0.999 & 0.45 &  & 1.000     & 0.48 &  & 0.999 & 0.59 \\
12    & 104  & 26   & 0.60                                    &  & 0.997 & 0.23 &  & 0.999 & 0.11 &  & 0.995 & 0.20  \\
13   & 109  & 27   & -0.15                                  &  & 0.991 & 0.33 &  & 0.997 & 0.56 &  & 0.999 & 0.70  \\
14   & 109  & 27   & -0.08                                  &  & 0.994 & 0.48 &  & 0.997 & 0.83 &  & 1.000     & 1.00    \\
15    & 111  & 28   & -0.16                                  &  & 0.995 & 0.38 &  & 1.000     & 0.56 &  & 1.000     & 0.73 \\
16    & 112  & 28   & -0.38                               &  & 1.000     & 0.54 &  & 1.000     & 0.55 &  & 1.000     & 0.71 \\
17   & 117  & 29   & 0.30                                    &  & 0.994 & 0.34 &  & 1.000     & 0.43 &  & 0.999 & 0.60  \\
18   & 119  & 30   & 0.28                                   &  & 0.995 & 0.42 &  & 0.997 & 0.55 &  & 0.999 & 0.75 \\
19   & 123  & 31   & -0.14                                  &  & 0.995 & 0.59 &  & 0.999 & 0.89 &  & 1.000     & 1.16 \\
20   & 136  & 34   & 0.51                                   &  & 0.998 & 0.36 &  & 0.999 & 0.31 &  & 0.999 & 0.52 \\
21   & 137  & 34   & -0.39                                  &  & 0.994 & 0.46 &  & 0.997 & 0.51 &  & 0.999 & 0.76 \\
22    & 147  & 37   & 0.55                                   &  & 0.991 & 0.27 &  & 0.996 & 0.24 &  & 0.998 & 0.48 \\
23    & 159  & 40   & 0.19                                   &  & 0.995 & 0.43 &  & 0.998 & 0.63 &  & 1.000     & 0.98 \\
24$^*$   & 189  & 47   & -1.15                                  &  & 0.905 & 0.18 &  & 0.852 & 0.3  &  & 0.938 & 0.12 \\
25   & 190  & 48   & -0.02                                  &  & 0.991 & 0.39 &  & 0.998 & 0.55 &  & 1.000     & 0.95 \\
26   & 194  & 48   & 0.15                                   &  & 0.988 & 0.33 &  & 0.996 & 0.51 &  & 1.000     & 0.89    \\
27   & 196  & 49   & -0.98                                  &  & 0.959 & 0.03 &  & 0.934 & 0.15 &  & 0.980  & 0.08 \\
28   & 199  & 50   & 0.62                                   &  & 0.986 & 0.19 &  & 0.991 & 0.16 &  & 0.999 & 0.46 \\
29   & 199  & 50   & 0.39                                   &  & 0.989 & 0.31 &  & 0.993 & 0.36 &  & 0.999 & 0.73 \\
30    & 200  & 50   & -0.08                                  &  & 0.997 & 0.37 &  & 0.998 & 0.58 &  & 1.000     & 0.99 \\
\bottomrule
\end{tabular}
\end{table}

Tables \ref{tab3} and \ref{tab4} show the empirical design-coverage and the design-average relative length of the intervals for the settings with variances $\dsiga^2=4$  and $\dsige^2=100$ when $e_{ij}$ has a normal distribution and $\alpha_{i}$ has either a normal distribution or a mixture distribution; settings with non-normal $e_{ij}$ are included in the supplementary material.  The areas are again presented and labeled in order of increasing size.  The  Monte Carlo standard errors for the design-coverage probabilities are approximately less than 0.01.
		
In all settings, the model holds and the within cluster variances are constant.  Nonetheless, when $\sigma_e^2/\sigma_a^2$ is large (as in Tables \ref{tab3} and \ref{tab4}), the kind of results we saw in the previous simulation occur.  In Table \ref{tab3}, there are 2 Group 3 areas, numbers 8 ($\hat{\alpha}_{8}/\hat\sigma_{\alpha}= -1.45$,\, $n_{8}=40$) and  15 ($\hat{\alpha}_{15}/\hat\sigma_{\alpha}= 1.13$,\, $n_{15}=26$), and 3 Group 2 areas,  numbers 10 ($\hat{\alpha}_{10}/\hat\sigma_{\alpha}= 1.03$,\, $n_{10}=44$),\, 20 ($\hat{\alpha}_{20}/\hat\sigma_{\alpha}= -0.91$, and $n_{20}=34$), 22 ($\hat{\alpha}_{22}/\hat\sigma_{\alpha}= -1.25$,\, $n_{22}=40$).
In Table \ref{tab4}, there are 2  Group 3 areas, numbers 1 ($\hat{\alpha}_{1}/\hat\sigma_{\alpha}= -1.27$,\, $n_{1}=20$), 24 ($\hat{\alpha}_{24}/\hat\sigma_{\alpha}= -1.15$,\, $n_{24}=47$), and 3 Group 2 areas, numbers 8 ($\hat{\alpha}_8/\hat\sigma_{\alpha}= -0.99$,\, $n_8=46$), 9 ($\hat{\alpha}_{9}/\hat\sigma_{\alpha}= 0.91$,\, $n_{9}=48$) and 10 ($\hat{\alpha}_{10}/\hat\sigma_{\alpha}= 1.04$,\, $n_{10}=48$).  On the other hand, when $\dsige^2/\dsiga^2$ is not large, there are no Group 2 or 3 areas.  Overall, we see that when $\dsige^2/\dsiga^2$ is large, areas with extreme EBLUPs and small to moderate sample sizes are more difficult than other areas to estimate well inthe design-based framework.

Why does the size of the random effect for an area matter in the design-based framework but not in the model-based framework?  In the model-based framework, the population and hence the random effect for an area is generated anew for each replication.  This means that in each sample we are estimating a realisation of an independent and identically distributed random variable so, over model-based replications, we are estimating the expected value of the random variable which is zero under the model.  
In the design-based framework, the population is fixed and the replication is over independent samples from this fixed population. Once generated, the random effects are fixed.  This means that in areas with extreme random effects, the EBLUPs are estimating the expected values of extreme order statistics which are not zero and difficult to estimate.  This then flows through into estimating the area mean of the survey variable for areas with large random effects.  In our design-based simulation, we used the population EBLUPs to assess the difficulty in estimation, but in practice, as shown in \cite{lyu2021asymptotics}, we would use the sample EBLUPs to estimate the population random effects.

The above discussion suggests that when we want to achieve good (model-assisted) design-based rather than model-based performance,  we should treat the random effects in the model as fixed.  
To check this intuition (and indirectly confirm the argument above), we repeated our design-based simulations treating $\alpha_i$ as fixed and examined the empirical design coverage and the relative design-expected length of the model-based prediction intervals constructed under the fixed area effects model. 
Explicitly, the general form of the approach is to rewrite the model (\ref{nerm}) and (\ref{mean}) with fixed area effects as a regression model
\begin{equation}\label{fixed}
\yij=\sz_{ij}^{*T}\boldsymbol\chi +\eij, \qquad\text{for $ j=1,\ldots,N_i, \, i=1,\ldots,g$,}
\end{equation}
where  $\sz_{ij}^{*}=[\su_i^T, \sx_{ij}^{(w)T},\sv_i^T]^T$ with $\sv_i$ a $g$-vector of zeros with a one in position $i$, and $\boldsymbol\chi=[\bxi^T,\bbeta_{2}^T,\alpha_1,\ldots,\alpha_g]^T$ with $\sum_{i=1}^g \alpha_i=0$.
The normal maximum likelihood estimator $\hat{\boldsymbol\chi}$ of $\boldsymbol\chi$ is the (constrained) ordinary least squares estimator  which can be computed using \texttt{lm} in R with the sum to zero constraint on the $\alpha_i$.  The optimal model-based predictor of $\bar{y}_i$ under (\ref{fixed}) is the composite estimator 
\begin{displaymath}
		\bar{y}_i^{\text{com-fixed}}=(1-k_i)\bar{y}_{i(s)}+k_i \bar{\sz}_{i(r)}^{*T}\hat{\boldsymbol\chi},
\end{displaymath}
where $\bar{\sz}_{i(r)}^{*}=[\su_i^T,\bar{\sx}_{i(r)}^{(w)T},\boldsymbol{\nu}_i^T]^T$ and an approximate $100(1-\epsilon)\%$ prediction interval for $\bar{y}_i$ is 
\[
\Big[\bar{y}_i^{\text{com-fixed}} - \Phi^{-1}(1-\epsilon/2) k_i\left\lbrace \frac{\hsige^2}{N_i-n_i}+\bar{\sz}_{i(r)}^{*T}\hat\bv_{\boldsymbol\chi}\bar{\sz}_{i(r)}^*\right\rbrace^{1/2},\, \bar{y}_i^{\text{com-fixed}} + \Phi^{-1}(1-\epsilon/2) k_i\left\lbrace \frac{\hsige^2}{N_i-n_i}+\bar{\sz}_{i(r)}^{*T}\hat\bv_{\boldsymbol\chi}\bar{\sz}_{i(r)}^*\right\rbrace^{1/2}\Big],
\]
where $\hsige^2$ estimates $\sigma_e^2$, and $\hat\bv_{\boldsymbol\chi}$ estimates the variance of $\hat{\boldsymbol\chi}$.
Similarly,  the synthetic estimator is $\bar{y}^{\text{syn-fixed}}_i=\bar{\sz}_{i}^{*T}\boldsymbol{\chi}$, where $\bar{\sz}_{i}^{*}=[\su_i^T,\bar{\sx}_{i}^{(w)T},\boldsymbol{\nu}_i^T]^T$, and an approximate $100(1-\epsilon)\%$ prediction interval for $\bar{y}_i$ is 
\[
\Big[\bar{y}_i^{\text{com-fixed}} - \Phi^{-1}(1-\epsilon/2) (\bar{\sz}_{i}^{*T}\hat\bv_{\boldsymbol\chi}\bar{\sz}_{i})^{1/2},\, \bar{y}_i^{\text{com-fixed}} + \Phi^{-1}(1-\epsilon/2)(\bar{\sz}_{i}^{*T}\hat\bv_{\boldsymbol\chi}\bar{\sz}_{i})^{1/2}\Big].
\]

The full set of results is included in the Supplementary material. The design-coverage results for both methods are generally good; the relative design-expected lengths of the intervals for the composite approach are mostly smaller than those for the synthetic approach.

\section{Discussion }\label{sec:discussion}

In this paper, we considered model-based small area estimation under the nested error regression model.  We discussed two target characteristics of interest, the small area means and the conditional linear predictors of the small area means, and the construction of mixed model estimators (EBLUPs) of these two targets.  We established asymptotic linearity results and central limit theorems for these estimators which allow us to establish asymptotic equivalences between estimators,  to approximate their sampling distributions, obtain simple expressions for and construct simple estimators of their asymptotic mean squared errors, and justify asymptotic prediction intervals.  Our new results are established under the asymptotic framework of increasing numbers of small areas and increasing numbers of units in each area, a framework that has not previously been applied in small area estimation.  We report model-based simulations that show that these results are applicable in quite small, finite samples, establishing that they fill important theoretical  gaps and are useful in practice.  In particular, our mean squared error estimator performs as well or better than the widely-used \cite{prasad1990estimation} estimator and is much simpler, so it is easier to interpret and consequently provides more insight.   We also carried out a design-based simulation using real data on consumer expenditure on fresh milk products to explore the design-based properties of the mixed model estimators.  This simulation produced some surprising results which we managed to explain and interpret through analysis of the population and further design-based simulations.  The simulations together highlight under-appreciated differences between the model- and design-based properties of mixed model estimators in small area estimation.

Given the extensive literature on small area estimation, it is important to acknowledge that in this paper we have considered only one of many interesting and important scenarios.  Future work should include applying the asymptotic approach to area level models, outlier robust estimators and the extensions to the basic nested error regression model discussed in the Introduction.

	\bibliographystyle{agsm}
	\bibliography{mythesisbib}
	
\end{document}